\documentclass[sigconf, nonacm]{acmart}

\AtBeginDocument{%
  }

\setcopyright{acmlicensed}
\copyrightyear{2024}
\acmYear{2024}
\acmDOI{XXXXXXX.XXXXXXX}

\acmConference[Conference acronym 'XX]{the ACM on Management of Data}{June 03--05,
  2024}{Woodstock, NY}
\acmISBN{978-1-4503-XXXX-X/18/06}

\usepackage{graphicx}
\usepackage{textcomp}
\usepackage{xcolor}
\usepackage{booktabs}
\usepackage{cleveref}
\usepackage{float}
\usepackage{bm}
\usepackage{tabularx}
\usepackage{graphicx}
\usepackage{subcaption}
\usepackage{multirow}
\usepackage{weiwAlgorithm}
\usepackage{diagbox}
\usepackage{threeparttable}
\usepackage{adjustbox}
\usepackage{url}
\usepackage{tabu}
\usepackage{array}
\usepackage{makecell}

\newcommand{\myparagraph}[1]{\vspace{0.5mm} \noindent \textbf{#1}.}
\newcommand{\myparagraphunder}[1]{\vspace{0.5mm} \noindent \underline{#1}.}
\newcommand{\myblue}[1]{{\color{black} #1}\xspace}

\newtheorem{definition}{Definition}[section]
\newtheorem{example}{Example}[section]

\begin{document}
\pagestyle{plain}
\pagenumbering{arabic}

\title{Deep Overlapping Community Search via Subspace Embedding} 







\author{Qing Sima$^{1}$, Jianke Yu$^{2}$, Xiaoyang Wang$^{1}$, Wenjie Zhang$^{1}$, Ying Zhang$^{2}$, Xuemin Lin$^{3}$}
\affiliation{\vspace{2mm}
    \institution{$^{1}$University of New South Wales, Sydney, Australia \\
    $^{2}$University of Technology Sydney, Sydney, Australia 
    $^{3}$Shanghai Jiao Tong University, Shanghai, China}
}
\email{
    {q.sima, xiaoyang.wang1, wenjie.zhang}@unsw.edu.au, 
    {jianke.yu@student, ying.zhang}@uts.edu.au, xuemin.lin@sjtu.edu.cn}

\begin{abstract}
Overlapping Community Search (OCS) identifies nodes that interact with multiple communities based on a specified query.
Existing community search approaches fall into two categories: algorithm-based models and \myblue{ML-based} models.
Despite the long-standing focus on this topic within the database domain, current solutions face two major limitations: 1) Both approaches fail to address personalized user requirements in OCS, consistently returning the same set of nodes for a given query regardless of user differences.
2) Existing \myblue{ML-based} CS models suffer from severe training efficiency issues. 
In this paper, we formally redefine the problem of OCS.  
By analyzing the gaps in both types of approaches, we then propose a general solution for OCS named \textbf{\underline{S}}parse \textbf{\underline{S}}ubspace \textbf{\underline{F}}ilter (SSF), which can extend any \myblue{ML-based} CS model to enable personalized search in overlapping structures.
To overcome the efficiency issue in the current models, we introduce \textbf{\underline{S}}implified \textbf{\underline{M}}ulti-hop Attention \textbf{\underline{N}}etworks (SMN), a lightweight yet effective community search model with larger receptive fields. 
To the best of our knowledge, this is the first \myblue{ML-based} study of overlapping community search. 
Extensive experiments validate the superior performance of SMN within the SSF pipeline, achieving a $13.73\%$ improvement in F1-Score and up to $3$ orders of magnitude acceleration in model efficiency compared to state-of-the-art approaches.
\end{abstract}
\maketitle

\section{Introduction}

\label{Introduction}
\begin{figure}[t]
  \centering
  \includegraphics[width=0.45\textwidth]{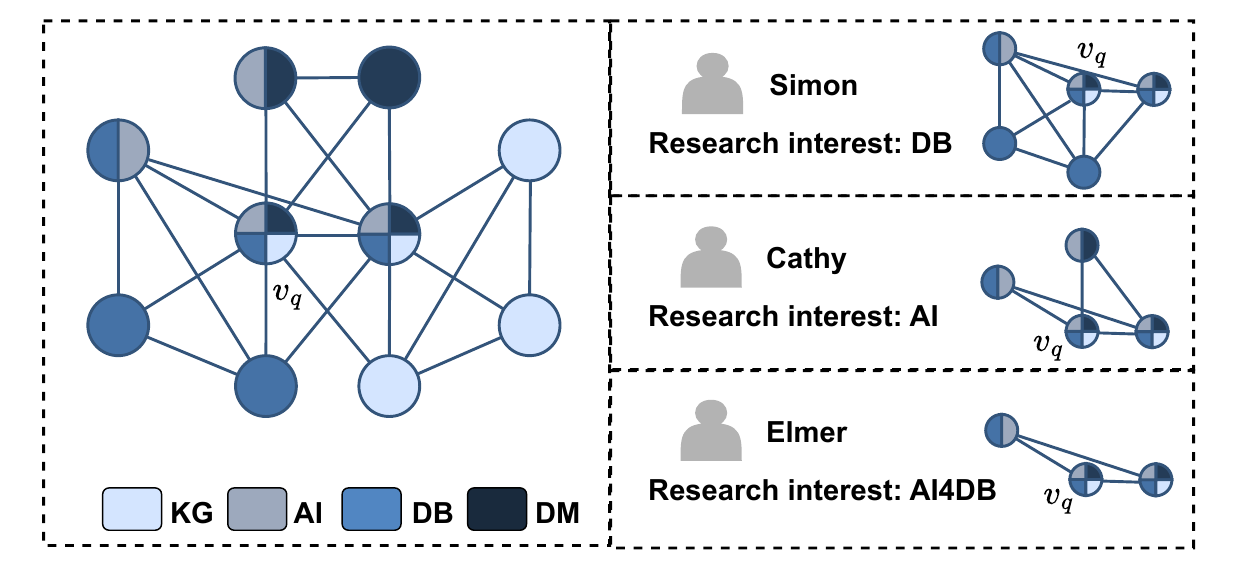}
  \vspace{-1mm}
  \caption{Different users are expecting different communities given the same query node}
  \label{fig:user}
  \vspace{-1mm}
\end{figure}

Identifying a closely interrelated community based on a query node is a long-standing focus within the database domain, facilitating various applications, including fraud detection~\cite{fang2020survey,coclep} and recommender systems~\cite{temporal_cs_2023,ics-gnn}. 
Existing Community Search (CS) models can be categorized into algorithm-based and \myblue{ML-based} approaches. 
Algorithm-based models define a community as a cohesive group of nodes~\cite{k-core,spatialcs_2015,spatialcs_2018}, while \myblue{ML-based} approaches are task-driven which define communities using labels, or node types~\cite{ics-gnn,temporal_cs_2023,meta_cs_2023,wang2024efficient}. 
Leveraging predictive capabilities, \myblue{ML-based} models identify a set of nodes closely related to the query sharing the same label~\cite{ics-gnn}.
However, these models cannot be extended to Overlapping Community Search (OCS). 
Therefore, this paper aims to develop efficient \myblue{ML-based} community search models for overlapping community structures, addressing the limitations of existing methods.

Overlapping community structure allows each node to interact with multiple communities, each exhibiting distinct characteristics such as sizes, levels of cohesiveness, and attribute patterns~\cite{2016_topk_ocs,k-clique,k-clique2018}. 
\myblue{This complexity raises a key challenge in OCS: how do we prioritize or rank communities when a query node belongs to multiple at once? 
Given the diversity in user interests, it becomes crucial to offer users the flexibility to personalize their search by selecting which target communities they wish to focus on. 
Existing methods often fall short by returning the same community for a given query, ignoring the differences in user preferences. 
Therefore, enabling the interactive selection of target communities is essential to ensuring the results align with user requirements.}

\autoref{fig:user} shows a toy example of citation networks, where nodes represent papers and edges denote citation relationships. 
The colors on the nodes indicate community affiliations, with multi-colored nodes representing papers overlapping multiple domains.
Community search can serve as a tool that recommends related papers based on a user's current reading, i.e., the query node.
The right part of the figure demonstrates how different users with distinct research interests seek personalized community recommendations based on the same query node $v_q$.
Simon, Cathy, and Elmer have different research interests: DB, AI, and AI4DB, respectively. 
Each subgraph highlights a set of nodes belonging to a domain, represented by the underlying color.
For example, users interested in the DB domain, like Simon, will specifically search for blue nodes.
Particularly, Elmer's interest in AI4DB represents an intersection between multiple communities, requiring the model to identify nodes that fall into both target communities: AI and DB. 
Hence, it is essential to allow users to select target communities addressing their personalized needs.  
Compared to disjoint community search~\cite{ics-gnn,qdgnn,coclep}, Overlapping Community Search (OCS) is a more challenging problem and has not received much attention in existing \myblue{ML-based} literature.
The study of OCS yields potential benefits across various applications.
For example, it enables the precise extraction of fraudulent entities from multiple communities~\cite{fraud_jianke_2023}, the discovery of literature in the cross-domain~\cite{citation_cross_2023}, and the recommendation of products to the most valuable community~\cite{recommond_2023}. 

\myparagraph{Existing solutions}  
Popular algorithm-based approaches employ different structural constraints to measure subgraph cohesiveness, such as $k$-core~\cite{k-core,spatialcs_2015,spatialcs_2018}, $k$-truss~\cite{k-truss2015,k-truss2017, k-truss2022}, and $k$-clique~\cite{k-clique,k-clique2017,k-clique2018}. 
Algorithm-based overlapping community search models are designed to discover multiple subgraphs containing the query node, each meeting a certain level of cohesiveness requirements~\cite{k-clique,2016_topk_ocs,k-clique2018}. 
However, these models struggle to assign predictive labels like ``DB'' to each node, assuming all returned communities carry the same semantics~\cite{ics-gnn}. 
In addition, these approaches only capture linear attribute-wise patterns and tend to measure structural cohesiveness and attribute homogeneity independently~\cite{qdgnn}. 

Comparatively, \myblue{ML-based} community search models are task-orientated and identify communities by prior knowledge learned from limited labels~\cite{ics-gnn,qdgnn,coclep, wang2024neural}. 
Current models are mainly built with Graph Neural Networks (GNNs) to learn node representations for online searching~\cite{gcn2016,gat_2017,graphsage,gin_2018}, such as ICS-GNN~\cite{ics-gnn}, QDGNN~\cite{qdgnn} and COCLEP~\cite{coclep}. 
\myblue{These models have two main phases, including offline training and online searching.
The offline training stage focuses on learning representative node embeddings from a subset of labeled data.}
Online searching algorithms utilize pairwise GNN scores (probabilities)~\cite{ics-gnn,qdgnn} or node similarity~\cite{coclep} against the query to identify communities.
However, current \myblue{ML-based} approaches primarily focus on disjoint community search, ignoring the nature that nodes tend to demonstrate various community affiliations. 
Therefore, two main motivations exist for designing an efficient and effective \myblue{ML-based} approach for OCS.

\begin{figure}[t]
  \vspace{-2mm}
  \centering
  \includegraphics[width=0.5\textwidth]{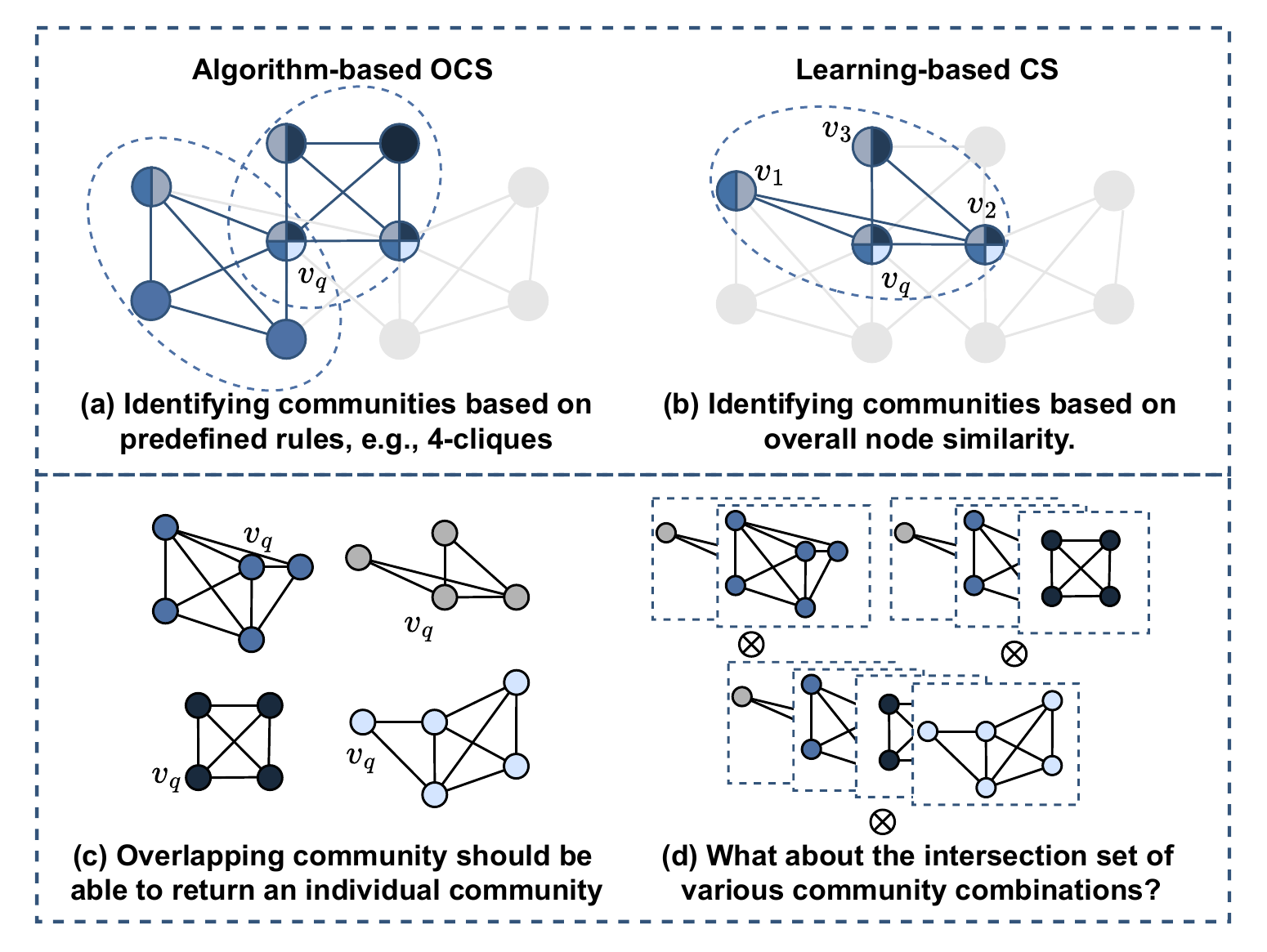}
  \vspace{-5mm}
  \caption{Challenges in existing approaches and user expectation given datasets with overlapping communities}
  \label{fig:challenges}
\vspace{-2mm}
\end{figure}

\myparagraph{Motivation 1} \textit{How to search for customized communities under overlapping community structures?} 
Although extensive work has been conducted, the aforementioned methods failed to address specified user requirements.
Given the same query, both types of models consistently return the same set of nodes for different users~\cite{ocs_kim2022dmcs,ics-gnn}.
\autoref{fig:challenges}(a) demonstrates the community identification process of an algorithm-based approach. 
\myblue{Given the query $v_q$ and a cohesiveness constraint, such as 4-clique~\cite{k-clique}, algorithm-based OCS models will return the highlighted communities by identifying two distinct subgraphs defined by 4-clique.}
This approach assumes that all cliques in the returned set share the same semantics, leading to two issues.
First, the models fail to distinguish nodes from different cliques, requiring manual intervention.
Additionally, cohesive subgraphs are unaware of downstream labels, and since not all nodes in the clique share the same label, the models struggle to exclude irrelevant nodes.

\myblue{ML-based CS models measure the overall node similarity guided by node labels~\cite{ics-gnn,qdgnn,coclep}. 
As illustrated in \autoref{fig:challenges}(b), given the query $v_q$, it tends to return nodes that exhibit higher similarity levels ($v_1, v_2, v_3$).
It is observable that the returned nodes will have closer embeddings by sharing at least two common communities with the query (each returned node shares at least two colors). }
To control the size of the returned set, these models utilize a threshold to gate the level of similarity~\cite{qdgnn, coclep}.
In this example, a relaxed threshold might include the entire graph, as all nodes share at least one label with the query shown in \autoref{fig:user}. 
Thus, both algorithm-based OCS and \myblue{ML-based} CS failed to identify user-specified communities.

\myparagraph{Motivation 2} \textit{How to design an efficient and effective model framework for \myblue{ML-based} OCS?}
Current \myblue{ML-based} models suffer from severe training efficiency issues.
COCLEP~\cite{coclep} uses graph partition techniques to improve the model training efficiency.
However, this approach has drawbacks such as loss of global context, boundary issues, and data imbalance.
Moreover, the complexity of overlapping community structures necessitates a model that favors high-order awareness, requiring it to gather messages from large receptive fields.
Therefore, designing an efficient model framework with larger receptive fields remains a significant challenge.

\myparagraph{Define ML-based OCS} 
By identifying challenges when applying the existing model to overlapping communities, we redefine the OCS problem within the deep learning context. 
A well-trained OCS model should account for each user's specific interests, guiding the community search to identify personalized communities.
\myblue{Therefore, an ML-based OCS model needs first to predict the potential labels for the query node.
It then allows users to interactively select the relevant labels as the target, thereby restricting the search to nodes within the specified target community to fulfill user requirements.}

In contrast to the expected community identified by existing approaches illustrated in \autoref{fig:challenges}(a) and 2(b), a qualified OCS model should be able to effectively retrieve those four pure communities individually, as depicted in \autoref{fig:challenges}(c). 
For example, the model should only return the blue node-set representing papers related to the DB domain, i.e., the target community. 
\myblue{This further raises a more challenging question, named Overlapping Communities Intersection Search (OCIS), illustrated in \autoref{fig:challenges}(d). 
What if a user is interested in multiple domains?

For example,  in citation networks (e.g., \autoref{fig:user}), a researcher might seek papers that lie at the intersection of multiple prominent fields, such as AI and DB. 
Given the vast number of papers in both fields, retrieving all papers from each domain would be overwhelming. 
The target is to narrow the search to find papers that cover both fields, meeting the researcher’s need for cross-domain insights and resulting in a more focused and relevant set of results.
Similarly, in social networks, community intersection search helps identify users with overlapping interests or affiliations. 
For instance, if a user engages in both technology and entrepreneurship communities, they can serve as a query to find others with similar profiles. 
When explicit labels are not available, the model can use predicted memberships from shared connections and attributes to identify relevant users, enabling more personalized recommendations and deeper insights into the network.


To solve OCIS, a possible brute-force approach is first to identify all the target communities separately and then calculate their intersections.
However, this transforms the task into an exhaustive enumeration of communities, which is impractical due to the high computational cost.}
This paper proposes a general solution to tackle the OCS problem.
The technique developed is also effective for the OCIS scenario without the requirement of enumerating all communities, reducing the computational overhead.

\myparagraph{Our solution} 
In this paper, we introduce a subspace community embedding technique called Sparse Subspace Filter (SSF) to tackle the challenge of identifying and segregating nodes with overlapping community affiliations. 
SSF is a general technique that can extend any existing \myblue{ML-based} model primarily built for the disjoint community search problem to OCS.
Moreover, to address the limitations in previous approaches, we replace the existing model with a novel framework named Simplified Multi-hop Attention Network (SMN), which significantly improves the model training speed while preserving high-order awareness. 

\myparagraphunder{Sparse Subspace Filter} 
SSF aims to learn a sparse matrix, representing each community by a sparse embedding.
This technique enables node embeddings to fall into multiple subspaces simultaneously, effectively identifying the target set under overlapping community structures. 
When searching for a target community, the learned sparse community embeddings are used as a basis vector to project nodes into the underlying subspace. 
The community search is then conducted exclusively within the target subspace. 
The proposed technique also extends any \myblue{ML-based} CS model to the scenario with multiple target communities, using the union subspace to represent the intersection of communities.  
SSF efficiently identifies the intersection between communities without the need to enumerate the entire graph for community affiliations.

\myparagraphunder{Lightweight model framework, SMN} 
To address the challenges in existing \myblue{ML-based} CS models, we propose a novel model named SMN. 
By proving the training inefficiency in the popular graph-query frameworks, SMN adopts a simplified model structure to alleviate the burden.
Moreover, as an OCS model needs a larger receptive field to capture high-order patterns, SMN uses an advanced hop-wise attention mechanism to cover higher-hop neighborhoods while preventing the model from oversmoothing. 

\myparagraph{Contributions} The main contributions of this paper are summarized as follows:
\begin{itemize}
    \item To the best of our knowledge, we are the first to investigate the problem of overlapping community search in the \myblue{ML-based} scenario. 
   
    \item 
    A general solution, named SSF, is then proposed, which is effective in finding a pure community as well as handling the intersection scenario. 
    
    \item 
    Moreover, we introduce a Simplified Multi-hop Attention Network (SMN), which is efficient in model training while capturing high-order patterns.
    
    \item
    Extensive experiments on $9$ overlapping and $4$ disjoint community datasets show that our model can achieve an average F1-Score improvement over state-of-the-art methods of $13.73\%$ and $7.62\%$, respectively. Additionally, our approach enhances model training efficiency by $3$ orders of magnitude and online query efficiency by $2$ orders of magnitude.
\end{itemize}


\section{related work} \label{related work} 


\myblue{\myparagraph{Algorithm-based community search} The community search problem is widely studied in the literature and can find many applications.
Different cohesiveness metrics are leveraged, such as $k$-core~\cite{k-core,spatialcs_2015,tan2023higher}, $k$-truss~\cite{k-truss2022,k-truss2015,k-truss2017}, and $k$-clique~\cite{k-clique,k-clique2017,k-clique2018}, which efficiently identify communities based on the graph structure.
Moreover, researchers conduct studies on attributed graphs and extend their analysis by incorporating attribute constraints alongside structural considerations to identify a set of nodes with similar attributes~\cite{attributed,attributed2,attributed3,attributed4}. 
Additionally, several studies have focused on discovering communities that contain multiple query nodes~\cite{k-core,k-truss2015,k-truss,k-truss2022}.
Given a set of query nodes, the studies aim to find a densely connected subgraph that contains all the query nodes. 
For example, CTC~\cite{k-truss2015} and FirmTruss~\cite{k-truss2022} are designed to search for the community including all the query nodes while satisfying different constraints, e.g., closest truss and firm truss.  
These studies are orthogonal to our research, focusing on overlapping communities and personalized community search.
Algorithm-based OCS~\cite{k-clique,2016_topk_ocs,k-clique2018} enable the query node to possess multiple community affiliations with equivalent levels of cohesiveness, such as being part of two subgraphs that fulfill $k$-clique constraints.
However, these models often return all communities containing the query nodes without the ability to focus on a specific community. 
Nonetheless, the lack of label awareness hampers these models' capacity to identify and separate nodes from distinct communities.}

\myblue{\myparagraph{GNN-based community search} GNN and its variants have achieved considerable success in graph analytic tasks, including node classification~\cite{gcn2016,graphsage} and subgraph mining~\cite{gog_2022,neursc_2022}. 
The GNN model learns from predefined ground truth, effectively capturing patterns from node attributes while considering diverse graph structures simultaneously.
In addition, advanced models have been introduced to improve model expressiveness and efficiency~\cite{gat_2017, gin_2018, sgc_2019}.
Recently, GNN-based community search models have attracted increasing attention due to their flexible structure constraints and expressive power.
These models can distinguish nodes from different communities by balancing the contribution from both the topological structure and the nodes' attributes. 
Deep CS models are trained using prior knowledge, making their assumptions more realistic than traditional approaches. 
A community is identified by a group of nodes sharing similar patterns in attributes and topological structures. 
ICS-GNN~\cite{ics-gnn} introduces an online deep community search model using a vanilla GCN model. 
The model is transductive, conducting training and online querying within the identified candidate subgraph. 
QDGNN~\cite{qdgnn} employs an offline setting by training the model on a fixed training set and inferring the model onto the unseen test set. 
The model extends to an attributed community search by adopting an attribute encoder to identify a group of nodes that contain a set of attributes. 
ALICE~\cite{wang2024neural} focuses on attributed community search by combining a candidate subgraph extraction phase using density sketch modularity.
The model follows the query-graph encoder frameworks and adopts a cross-attention encoder to control the interaction. 
COCLEP~\cite{coclep} follows the framework of QDGNN and conducts semi-supervised training by leveraging contrastive learning techniques. 
The model uses a hypergraph as an augmented graph and propagates information using GCN and Hyper GNN~\cite{hgnn2019}. 
However, current models have struggled to adapt to the overlapping community search, often encountering issues with oversmoothing and slow training.}
\section{Preliminaries} \label{Preliminaries}


\subsection{Problem Definition}

Let $G=({V,E})$ be an undirected graph with a set ${V}$ of nodes and a set ${E}$ of edges. 
Let $n =  |{V}|$ and $m = |{E}|$ be the number of nodes and edges, respectively. 
Given a node $u \in {V}$, ${N}(u)=\{v|(u, v) \in {E} \}$ is the neighbor set of $u$. 
The adjacency matrix of $G$ is denoted as $\bm{A} \in \{0, 1\}^{n\times n}$, where $\bm{A}_{i,j} = 1 \text{, if } (v_i, v_j) \in {E}$, otherwise $\bm{A}_{i,j} = 0$. 
$\bm{X} = \{ \bm{x}_1,$ $\bm{x}_2,...,\bm{x}_n \}$ is the set of node features and $\bm{x}_i$ represents the node features of $v_i$. 
Given a query node $q$, the CS problem aims to find a $k$-sized set of nodes containing the query from $G$ while maximizing the GNN score or node similarity against the query~\cite{ics-gnn,qdgnn,coclep}. 
Under overlapping community structures, each node $u$ belongs to more than one community, i.e., $u\in {\mathcal{C}}_u = \{{C}^{z_1}_u,{C}^{z_2}_u, ..., {C}^{z_i}_u\}$, where ${\mathcal{C}}_u$ is the set of communities contains $q$, and $z_i \in Z$ is the label of a community. 
Users are allowed to select a target community label ${t} \in Z$ to guide the community search.
Following the existing definition of community search in \myblue{ML-based} models~\cite{ics-gnn,qdgnn,coclep}, we define OCS as below:

\begin{definition}[Overlapping Community Search, OCS]
    \label{def:1}
     Given a graph $G$, a query node $q$, a community size $k$, and a target community label ${t} \in Z$, OCS aims to identify a $k$-sized query-dependent group of nodes $V_c$ that are closely intra-related.
     This group satisfy $V_c \subseteq {C}^{t}_q$ and $|V_c| = k$, where ${C}^{t}_q$ is the true community associated with the target label $t$. 
     
\end{definition}

Under this definition, the user is only interested in a single community, i.e., the target community.
However, as discussed, due to the complexity of overlapping structures, describing the desired group using a single target community might not be adequate. 
Defining a more refined community by considering the intersection of multiple target groups is often preferable. 
Therefore, we extend OCS by introducing the following definition to enhance flexibility.

\begin{definition}[Overlapping Communities Intersection Search, OCIS]
    \label{def:2}
     Given a graph $G$, a query node $q$, a community size $k$, and multiple target community labels $\bm{T_q} = \{t_1,t_2,...,t_i\} \in Z$, OCIS aims to search for the user-specified intersection set $V_c'$ of size $k$ such that $V_c'  \subseteq \hat{\bm{C}}_q $, where $\hat{\bm{C}}_q = \bm{C}^{t_1}_q \cap \bm{C}^{t_2}_q \cap ... \cap \bm{C}^{t_i}_q$
\end{definition}

Within this definition, the intersection of multiple communities represents a refined community that is valuable to end users. 
It is worth mentioning that employing a brute-force approach, which involves enumerating all nodes for community prediction and then joining multiple communities to determine the intersection, is an impractical strategy. 
Therefore, we aim to efficiently identify the intersection without enumerating the dataset. 
\myblue{In the literature, some studies have focused on discovering communities that contain multiple query nodes, e.g.,~\cite{qdgnn,k-truss2015}. 
In datasets with overlapping labels, this problem can transition into a single-target or multi-target community search. 
If the input nodes share one or more community memberships, those communities become the targets. 
In cases where no common community exists, techniques such as majority voting can be used to identify the target community. 
As a result, the proposed OCS and OCIS methods can be seamlessly extended to scenarios involving multiple query nodes.}


\myblue{In this paper, the overlapping community search task operates in a semi-supervised framework. 
Commencing with a graph represented as $G({V,E})$, the model is trained on a small fraction of the dataset (10\% or less). 
Given a graph $G$, we aim to design a model and a search algorithm to handle both OCS and OCIS. The model should be efficient while capable of handling high-order patterns. }

\subsection{Graph Convolutional Network}
Graph Convolutional Network (GCN)~\cite{gcn2016} is the most commonly employed variant of GNN, leveraging a low-pass filter (the first-order adjacency matrix) to gather information solely from its neighbors rather than all local nodes. 
The propagation process is represented as \autoref{eq:1}:
\begin{equation} \label{eq:1}
\bm{H}^{(l+1)} = \sigma(\bar{\bm{D}}^{-\frac{1}{2}}\bar{\bm{A}}\bar{\bm{D}}^{-\frac{1}{2}}\bm{H}^l\bm{W}^l),
\end{equation}
\noindent where $\bm{H}^l$ is the hidden state from the $l$ layer, $\bar{\bm{D}}$ is the normalized degree matrix which is a diagonal matrix of node degree, $\bar{\bm{A}}$ is the adjacency matrix with self-loop, $\bm{W}^l$ is the learnable weight matrix and $\sigma$ is an activation function.
$\bar{\bm{A}}\bm{H}^l$ demonstrates how a node aggregates information from its one-hop neighbors.
The activation functions add non-linearity between layers and prevent multiple linear functions from collapsing into a single one. 
Common activation functions include $\text{ReLU}(\cdot)$, $\text{sigmoid}(\cdot)$, and $\text{LeakyReLU}(\cdot)$. 
Loss is computed by comparing the model output with the ground truth and using backpropagation to update model parameters iteratively. 

In contrast to other deep learning models, a deeper GCN does not enhance its expressiveness. 
With each additional layer, the receptive fields of the GCN expand by one hop.
Deeper models lead nodes to aggregate information from the entire graph, diminishing its ability to distinguish nodes, known as oversmoothing~\cite{deepgcns}.

\section{Subspace Community Embedding: A General Solution for OCS}\label{SSF}

This section introduces the subspace community embedding technique, a general solution extending any \myblue{ML-based} CS model to OCS.
We then propose the SMN in \Cref{sec:SMN} as the backbone model to solve existing models' challenges. 

Sparsity plays a crucial role in enhancing machine learning models across various applications.
It benefits various areas such as subspace clustering~\cite{ssc_2013_subspace,lrr_2012_subspace,kssc_2014_subsapce, deepssc_2017,efficient_dssc}, sparse training~\cite{sparss_nn_2017, sparse_training_2022}, and sparse feature selection~\cite{lemhadri2021lassonet,sparse_feature_2021}.  
Inspired by this concept, we introduce a subspace embedding technique named Sparse Subspace Filter (SSF).
SSF trains node embeddings to align closely with their corresponding community embeddings, minimizing Euclidean and cosine distances within the subspace.
This method effectively approximates community representations, enhancing the model's accuracy and relevance in OCS.
The sparse subspace filter is initialized as a trainable matrix with $(s, c)$ dimensions, where $s$ is the dimension of the output embeddings, and $c$ denotes the number of communities. 
In the following, we detail our approach by answering the following questions.

\noindent{\textbf{What roles does SSF play in our model?}} SSF plays two roles in the model, including a filter of the model classifier during the offline training and the basis matrix to guide the community search during the online searching. 

\myparagraphunder{Offline training} In the training phase, the model clusters nodes from the same community into a subspace, allowing the trained SSF to function as a basis matrix representing all community subspaces.
As illustrated in \autoref{fig:ssf_loss}(a), SSF represented by $\bm{S}\in \mathbb{R}^{s\times c}$ is a sparse matrix with elements drawn from a Bernoulli variable.
The black color demonstrates 1 at the underlying position, and the white denotes 0. 
By performing the element-wise product with the classifier matrix $\bm{W_c}$, it gates the weight in the classifier, promoting sparsity. 
Given node embeddings learned by a random model, the gated classifier linearly transforms the node embeddings $\bm{\mathcal{H}}_{s} \in \mathbb{R}^{n\times s}$ into the likelihoods of community affiliation. 
As the columns in the classifier approximate the embeddings for the corresponding community, SSF ensures that each community embedding is related to only a subset of elements in the node embedding.
Hence, SSF projects community embeddings into distinct subspaces.
This design facilitates overlapping structures as a node embedding with full space covering multiple subspaces simultaneously.  

\myparagraphunder{Online searching} During online searching, the underlying columns of SSF can map nodes into the user-selected subspace. 
\autoref{fig:ssf_loss}(b) illustrates how subspace mapping benefits the overlapping community search.
For example, if the target domain is the database (DB), node $v_1$ will be identified as a noise node as it does not have a blue color. 
However, $v_q$ and $v_1$ may appear similar in the full space because they share two common labels, making them indistinguishable.
By projecting all nodes into the subspace representing DB, $v_1$ is positioned far from $v_q$ because the elements of its embedding that do not relate to DB are converted to 0, effectively distinguishing it from nodes within that domain.
This enhances the model's ability to differentiate nodes based on community relevance.

\begin{figure}[t]
  \centering
  \includegraphics[width=0.5\textwidth]{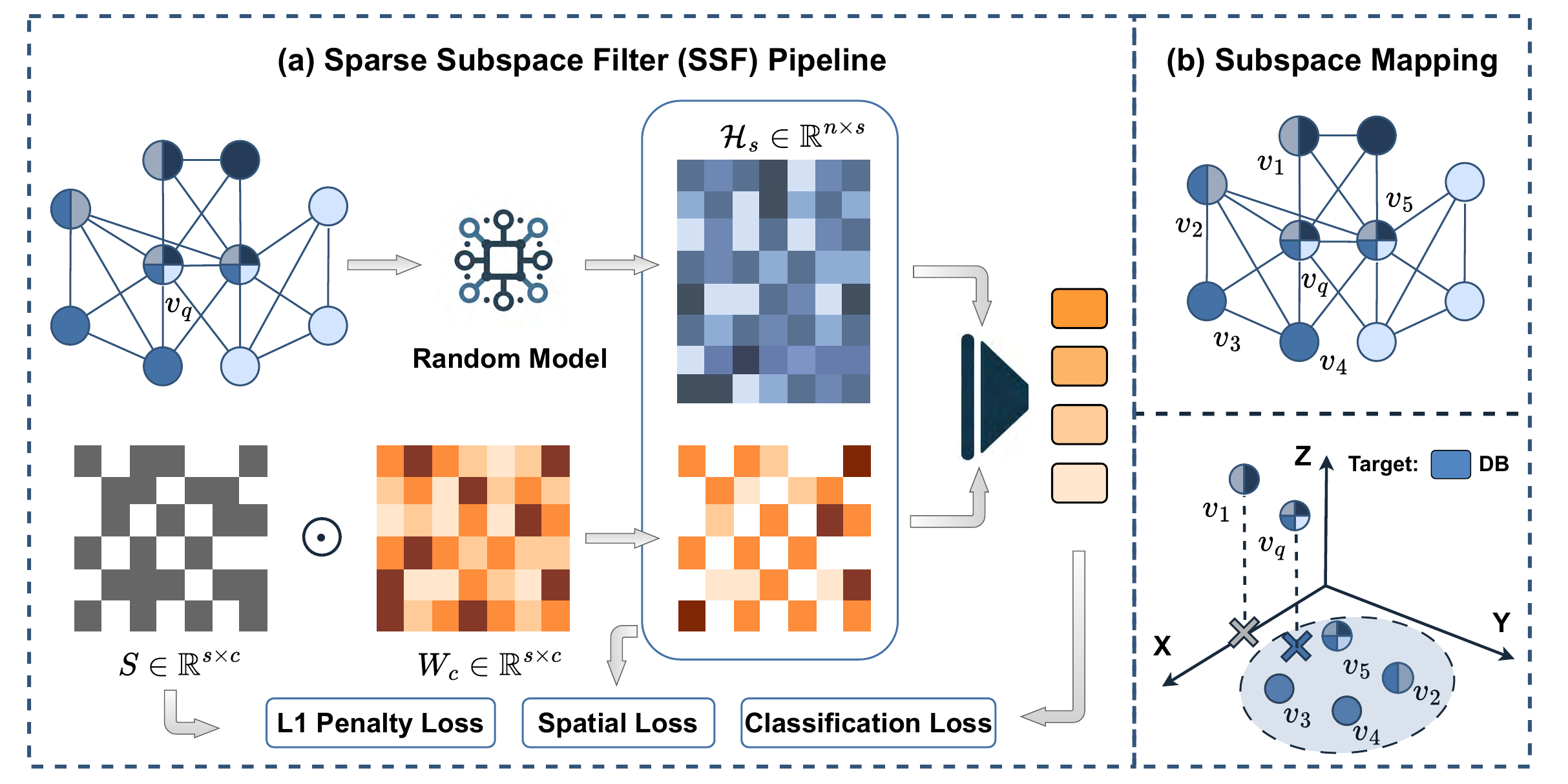}
  \vspace{-5mm}
  \caption{Subspace community embedding via the sparse subspace filter}
  \label{fig:ssf_loss}
\vspace{-2mm}
\end{figure}

\noindent{\textbf{How are the objective functions designed to train SSF?}} To make sure the SSF is well-trained to facilitate the overlapping community search, we adopted three objective functions.

\myparagraphunder{L1-penalty term} To induce sparsity while ensuring the objective function is differentiable, we train SSF as the real-valued parameters.
We then perform a maximum-likelihood (ML) draw by thresholding the values at 0.5 to sparse the SSF. 
The L1-penalty term is stated as \autoref{eq:l1_penalty} to ensure the model is in favor of a sparse SSF.
\begin{equation} \label{eq:l1_penalty}
\hat{\theta}, \hat{\Phi} = \arg \min_{\theta, \Phi} \left( \ell\left(y \mid \theta, \Phi, y \right) + \lambda \left\| \Phi \right\| \right), \ \left\| \Phi \right\| = \sum_{i=1}^{s} \sum_{j=1}^{c_i} S_{i,j},
\end{equation}
\noindent where $\theta$ and $\Phi$ are the parameters that minimize the loss function, and $\Phi$ is the penalty term regulate by $\ell_1$ on elements in SSF. 
The regularization term is scaled by $\lambda$ to control the level of sparsity.
The L1-penalty term ensures that the model is in favor of sparsity, facilitating the model to learn community embeddings falls in different subspaces. 
This ensures that large communities have a loose constraint in estimating node affiliations, suggesting a higher probability of demonstrating high similarity to node embeddings. 

\myparagraphunder{Classification loss} As mentioned in the definition, due to the query in OCS carrying various semantics, the OCS model should first predict the community affiliations, then allow the user to customize their target community.
Hence, we adopt a classification loss to supervise the model performance on community prediction.
Under overlapping structures, as each node denotes more than one community affiliation, the model tends to suffer the positive-negative imbalance issue. 
Where most nodes belong to a small fraction of the possible communities, implying the positive samples will be much less than the negative samples. 
To address this issue, we adopt the ASL loss~\cite{asl_loss_2021} to assign different exponential decay factors to positive and negative samples. 
A general form of a binary loss per label, $\mathcal{L}$, is given by \autoref{eq:class_loss}:
\begin{equation} \label{eq:class_loss}
\mathcal{L}_c = -y \mathcal{L}_+ - (1 - y) \mathcal{L}_-,
\end{equation}
\noindent where $\mathcal{L}_+$ and $\mathcal{L}_-$ are the positive and negative loss parts. 
Comparatively, ASL loss is defined as \autoref{eq:asl_loss}:
\begin{equation} \label{eq:asl_loss}
\begin{cases}
\mathcal{L}_+ = (1 - p)^{\gamma_+} \log(p) \\
\mathcal{L}_- = (p_m)^{\gamma_-} \log(1 - p_m),
\end{cases}
\end{equation}
\noindent where $p_m = \max(p - m, 0)$ is a shifted probability, monitoring $p$ to get $\mathcal{L}_- = 0$ when $p<m$, $(\gamma_+, \gamma_-)$ are focusing parameters for positive and negative samples, respectively. 
ASL balances the contribution from positive and negative samples through a soft threshold $(\gamma_+, \gamma_-)$ and a hard threshold (probability margin $m$). 
As depicted in \autoref{fig:ssf_loss}, the node embeddings $\bm{\mathcal{H}}_s$ perform a matrix multiplication with the gated classifier $\bm{W}_c'$ to generate the output logit. 
The classification loss is calculated by comparing the output logit with the ground truth labels.
This loss effectively groups the node embeddings from the same community.  

\begin{figure*}
  \centering
  \includegraphics[width=\textwidth]{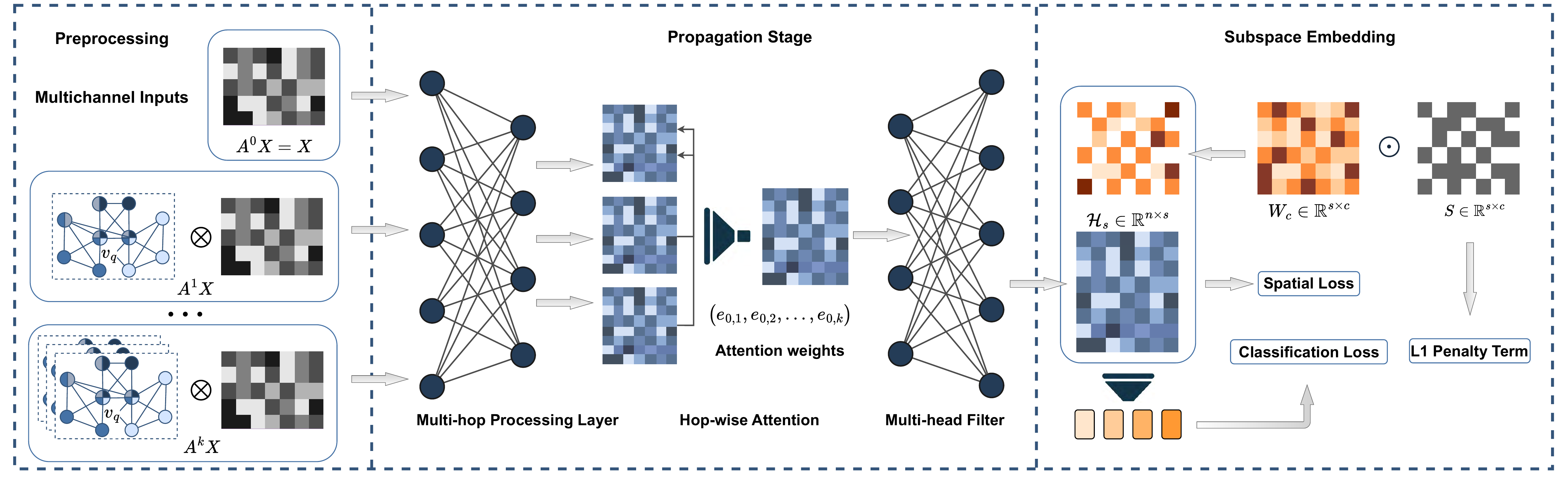}
  \vspace{-5mm}
  \caption{The architecture of SMN}
  \label{fig:SMN}
  \vspace{-3mm}
\end{figure*}

\myparagraphunder{Spatial loss} Furthermore, a spatial loss function is introduced to supervise the subspace mapping, ensuring nodes only fall into the subspaces representing their community affiliations. 
For nodes belonging to a community, their embeddings should be close to their community in the underlying subspace. 
Hence, two distance metrics are employed to monitor it, including Euclidean distance and cosine similarity. 
The loss is primarily generated on the non-zero elements related to each underlying subspace to accommodate overlapping community structures. 
Therefore, we filter node embeddings by the basis vector of each column in SSF before measuring the distance. 
Discrepancies against zero-like elements in SSF are not penalized, given their potential contribution to other communities. 
The sigmoid functions are then applied to the distance and similarity to derive the likelihood of nodes belonging to each community. 
This output is averaged into the final spatial distance as \autoref{eq:10}. 
\begin{equation} \label{eq:10}
\bm{D} = \frac{1}{2} (\sigma(-\text{dist}(\bm{h}_v, \bm{W}'_c)) + \sigma(\text{sim}(\bm{h}_v, \bm{W}'_c)),
\end{equation}
\noindent where ${\sigma}$ represents a sigmoid function, $\bm{W}'_c$ represents the gated model classifier, and $\bm{h}_v \in \bm{\mathcal{H}}_{s}$ is the final embeddings. 
$\text{dist}(\cdot)$ and $\text{sim}(\cdot)$ represent the Euclidean distance and the cosine similarity, respectively.
Similar to the classification loss, we compute the ASL for the spatial loss $\mathcal{\bm{L}_{\text{s}}}$ against the ground truth as \autoref{eq:11}:
\begin{equation} \label{eq:11}
\mathcal{\bm{L}_{\text{s}}} = \sum_{v=1}^{\eta} ASL(\bm{d}_{v},\bm{y}_{v}),
\end{equation}
\noindent Therefore, the final loss function is defined as follows:
\begin{equation} \label{eq:12}
\mathcal{\bm{L}} = \frac{1}{2} (\frac {\mathcal{\bm{L}_{\text{c}}}}{\delta^2_{c}} + \frac {\mathcal{\bm{L}_{\text{s}}}} {\delta^2_{s}}) + \lambda \left\| \Phi \right\| ,
\end{equation}
\noindent where ${\delta}$ is a parameter the model trains to balance the above two loss functions. In the experiments, we observe that this fused loss function stabilizes the model performance.

\section{ Simplified Multi-hop Attention Network (SMN)}\label{sec:SMN}

This section elaborates on the design details of the proposed SMN, which is a lightweight model with large receptive fields.   
We first introduce the overall framework of SMN to establish a comprehensive understanding as depicted in \autoref{fig:SMN}. 
As the subspace community embedding is already illustrated in \autoref{fig:ssf_loss}, in this section, we present the model by mainly focusing on the model preprocessing and propagation phase.

\subsection{Framework}
\autoref{fig:SMN} presents the framework of SMN, which consists of three main components, including preprocessing, propagation, and subspace community embedding.
We prove that the existing widely-used query encoder is not gaining model expressive power but slowing down the training process, detailed analysis disclosed in \Cref{theoretical}.
Hence, we removed the query encoder for model efficiency. 
In addition, we adopted a simplified framework to further improve the training speeds. 
This framework removes activation functions between layers and aggregates multi-hop neighborhood messages during preprocessing instead. 

In the preprocessing stage, SMN generates multichannel inputs by stacking messages from different hops.
The $k$-th channel represents the feature matrix with $(k-1)$ hops awareness.
Aggregating neighborhood information during preprocessing eliminates the need for expensive message-passing during the model propagation.
Hence, the model training speeds are further accelerated.  

The propagation stage consists of three layers: a multi-hop processing layer, a hop-wise attention layer, and a multi-head filter layer.
The multi-hop processing layer inputs the original features from each hop to linearly transform the multi-hop messages. 
The messages are then fused into single-channel messages through a hop-wise attention layer. 
The resulting outputs are further transformed through a multi-head filter layer, yielding the final embeddings.
The hop-wise attention mechanism assigns decaying weights to messages from various hops based on their contributions. 
This design effectively addresses the oversmoothing issue, enhancing model high-order awareness to facilitate overlapping community search. 
The final embeddings are then fed into SSF, learning subspace community embedding for OCS.

\subsection{SMN: Preprocessing and Propagation}

\myparagraph{Preprocessing} The preprocessing stage can be split into aggregation and normalization. 

\myparagraphunder{Aggregation} Inspired by the works~\cite{sgc_2019, lightgcn_2020, selfloop_nips2021}, SMN removes the non-linear activation functions during aggregation to improve the model training speed.
As proved by~\citet{non-linearity_2022}, linear propagation performs similarly to non-linear propagation, especially when graph structures are more informative compared to node attributes. 
A two-layer GCN can be represented as \autoref{eq:2}:
\begin{equation} \label{eq:2}
\bm{Z} = \text{softmax}(\hat{\bm{A}} \times \text{ReLU}(\hat{\bm{A}}\bm{XW}^{(0)})\bm{W}^{(1)}),
\end{equation}
\noindent where $\bm{Z}$ is the final output, softmax is a classifier that maps the probability of nodes belonging to different classes. 
$\text{ReLU}(\cdot)$ is the activation function to provide nonlinearity to the model. 
$\hat{\bm{A}}$ denotes the degree normalized adjacent matrix, $\bm{X}$ is a matrix of node features and $\bm{W}$ is a learnable matrix. 
$\bm{W}^{0} \bm{W}^{1}$ represents the weight for different layers of the networks. 
By removing the activation functions, SMN can be represented as \autoref{eq:3}:
\begin{equation} \label{eq:3}
\begin{split}
\bm{Z} &= \text{softmax}(\hat{\bm{A}} \times (\hat{\bm{A}}\mathbf{X}\bm{W}^{(0)})) \\
&= \text{softmax}(\hat{\bm{A}}^{2}\bm{XW}^{(0)}).
\end{split}
\end{equation}
\myblue{Since the computation of $\hat{\bm{A}}^{2}\bm{X}$ is equal to a preprocessing step, the level of the parameter is the same as a logistic regression model.} 
However, this simplified framework faces an obvious limitation: by aggregating the neighborhood message, the node features will quickly become indistinguishable. 
Therefore, this limits the model to a relatively low-hop awareness and harms OCS. 

Instead of directly using $\hat{\bm{A}}^k\bm{X}$ as input, SMN iteratively stacks adjacency matrices from different hops $\hat{\bm{A}}^0, \hat{\bm{A}}^1,...\hat{\bm{A}}^k$, and generates multichannel inputs by assigning node features such as $\bm{X}, \hat{\bm{A}}\bm{X},...\hat{\bm{A}}^k\bm{X}$. 
The ${k}$-th channel represents the node feature matrix with a $k-1$ hop receptive field. 
By aggregating neighborhood information during preprocessing, SMN eliminates the necessity of a GCN layer but employs a fully connected layer instead. 

\myparagraphunder{Normalization} In degree normalization, the adjacency matrix of GCN is normalized as $\hat{\bm{A}} = \bar{\bm{D}}^{-\frac{1}{2}}\bar{\bm{A}}\bar{\bm{D}}^{-\frac{1}{2}}$, where $\bar{\bm{A}} = \bm{A} + \bm{I} $ and $\bm{I}$ is the identity matrix, representing each node in $\bm{A}$ to add a self-loop by $\bm{A} + \bm{I}$ before normalization. 
The self-loop avoids the loss of self-features during aggregation. 
In contrast, SMN specifically removes the self-loop for two reasons: First, SMN takes the input of $\tilde{\bm{A}}^0\bm{X}, \tilde{\bm{A}}^1\bm{X},...\tilde{\bm{A}}^k\bm{X}$, where $\tilde{\bm{A}}^0\bm{X} = \bm{X}$ is the initial features matrix which prevents the loss of the self-features. 
Furthermore, removing the self-loop reduces redundancy during message passing and further differentiates messages collected from each hop. 
Thus, degree normalization in SMN is depicted in \autoref{eq:4}:
\begin{equation} \label{eq:4}
\tilde{\bm{A}} = \bm{D}^{-\frac{1}{2}}\bm{A}\bm{D}^{-\frac{1}{2}},
\end{equation}
\noindent where ${\bm{A}}$ and $\bm{D}$ represent the adjacency and the degree matrices without self-loop. 

\begin{figure}[t]
  \vspace{-3mm}
  \centering
  \includegraphics[width=0.5\textwidth]{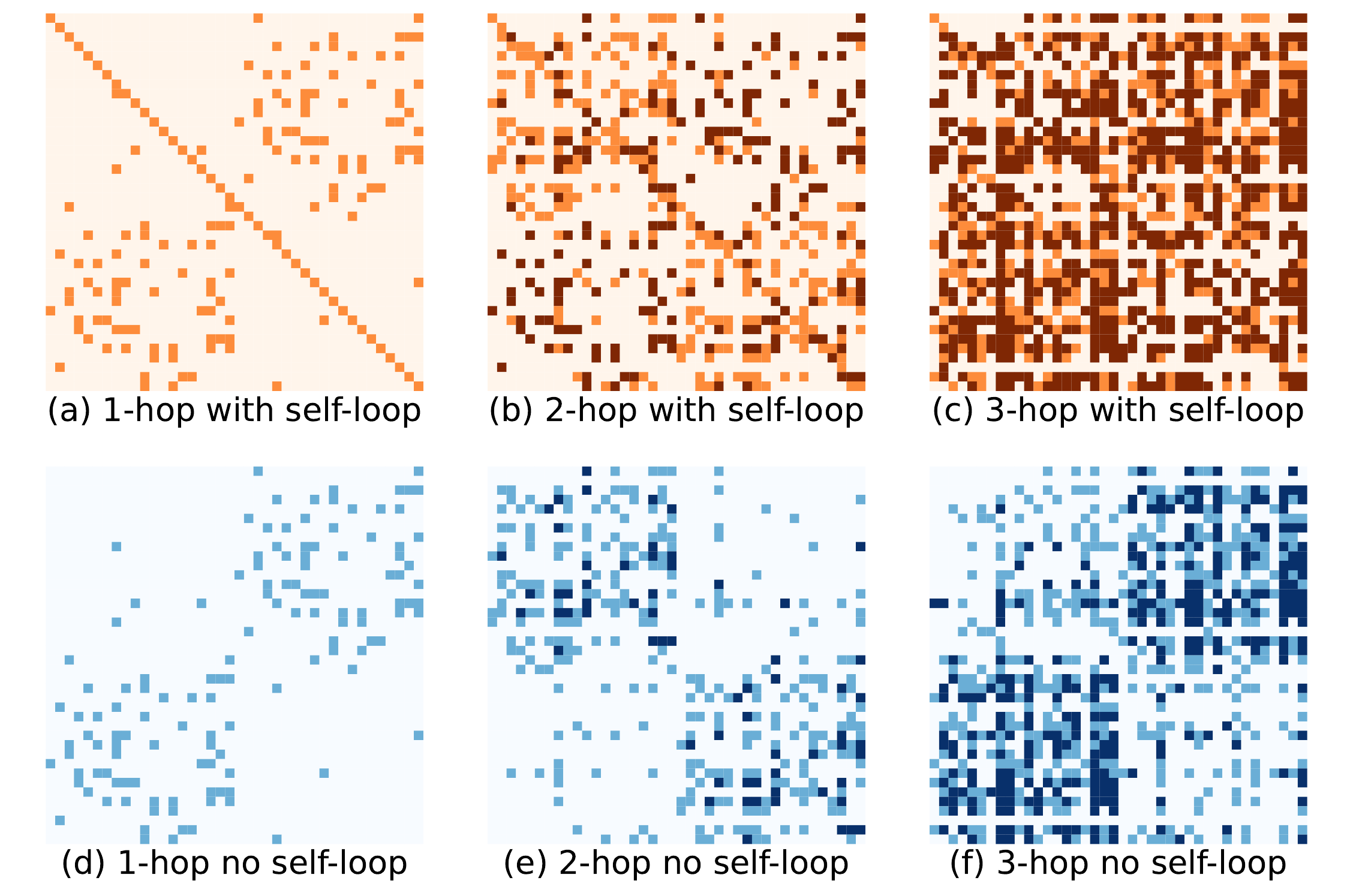}
\vspace{-6mm}
  \caption{Self-loop oversmooth messages received}
  \label{fig:self_loop}
\vspace{-2mm}
\end{figure}

\autoref{fig:self_loop} illustrates the comparison of adjacency matrices with and without self-loops. 
It shows that adding self-loops leads to a notable acceleration in graph exploration, causing oversmoothing within three hops.
Adversely, removing self-loops leads to a better contrast across the adjacency matrices among various hops.
It can be seen that the \autoref{fig:self_loop}(d) (1-hop no self-loops matrix) primarily focuses on direct neighbors, presented as the matrix's top-right and bottom-left corners. 
In contrast, the 2-hop matrix (\autoref{fig:self_loop}(e) emphasizes neighbors with a 2-hop distance (top-left and bottom-right corners), ignoring the 1-hop neighborhood. 
By comparing the \autoref{fig:self_loop}(c) and 5(f), removing self-loops effectively slows down the oversmoothing progress while enabling the proposed attention mechanism to capture unique patterns from various hops.

\myparagraph{Propagation: hop-wise attention} The hop-wise multi-head attention mechanism regulates the aggregation of messages from various hops, enabling the model to capture higher-order patterns while mitigating the oversmoothing effect.
SMN first applies a multi-hop processing layer to transform the initial features linearly to obtain sufficient expressive power. 
Here, the weight matrix $\bm{W}^l$ is shared across nodes and hops. 
We then perform self-attention on the hidden state to compute the attention coefficients as \autoref{eq:5}:
\begin{equation} \label{eq:5}
\begin{split}
\bm{e_i} &= {a}(\bm{W}^l\bm{H}_0, \bm{W}^l\bm{H}_i), \quad \forall i \in [0..k] \\
&= ({\overrightarrow{\bm{a}}}^T\bm{W}^l\bm{H}_0 + \overrightarrow{\bm{a}}^T\bm{W}^l\bm{H}_i) ,
\end{split}
\end{equation}
where $\bm{e_i}$ indicates the importance of $\bm{H}_i$ the ith-hop features toward $\bm{H}_0$ the zero-hop  (self-features matrix). $\overrightarrow{\bm{a}}$ is a shared attention mechanism $\overrightarrow{\bm{a}}\in \mathbb{\bm{R}}^{d'}$, and $k$ is the number of hops. 
\myblue{To fuse the message from different hops, the coefficients are first activated by a LeakyReLU~\cite{gat_2017}, which improves stability by allowing a small gradient for negative inputs, and then normalized by the softmax as \autoref{eq:6}.}
The obtained final weights $\bm{\alpha}$ weighted sum the multi-hop feature matrices into a single channel. 
\begin{equation} \label{eq:6}
\bm{\alpha}_{i} = \frac{\exp\Big(\text{LeakyReLU}\Big(\bm{e_i}\Big)\Big)}{\sum_{j\in [0..k]} \exp\Big(\text{LeakyReLU}\Big(\bm{e_j}\Big)\Big)}_.
\end{equation}

To further improve the performance, we observe that multi-head attention is beneficial in stabilizing the performance.
Similar to the graph attention networks~\cite{gat_2017}, multiple independent attention mechanisms are applied to the hidden state, and the output of each head is further concatenated into the final output. 
The model uses the multi-head filter $\bm{W}^r$ to fuse the output from different heads into final embeddings as \autoref{eq:7}:
\begin{equation} \label{eq:7}
\bm{\mathcal{H}}_s = \sigma\Big(\bm{W}^{r}\Big(\big\Arrowvert_{i=1}^{I}\sum_{k=0}^{{K}} \bm{\alpha}_{k}^i \bm{W}^l \bm{H}_k\Big)\Big),
\end{equation}
\noindent where $I$ is the number of heads and $K$ is the number of hops. 
The dimension of the final output $\bm{\mathcal{H}}_s$ is a hyper-parameter that matches the dimensions of the subspace community embeddings.

The algorithm for SMN propagation is presented in \autoref{alg:SMN}. 
\begin{algorithm}[t]
	{
            \footnotesize

		\caption{Preprocessing and SMN propagation}\label{alg:SMN}
		\Input{Feature matrix $\bm{X}$, the adjacency matrix $\tilde{\bm{A}}$, the number of hops ${k}$, the sparsity rate ${r}$}
		\Output{Model output $\mathcal{\bm{O}}$, final embeddings $\bm{\mathcal{H}}_{s}$, learned sparse subspace filter $\bm{S}$}
		\State{$\tilde{\bm{A}} = \bm{D}^{-\frac{1}{2}}\bm{AD}^{-\frac{1}{2}}$}
		\State{$\bm{\mathcal{H}} = \{\tilde{\bm{A}}^0\bm{X}, \tilde{\bm{A}}^1\bm{X},...,\tilde{\bm{A}}^{({k}-1)}\bm{X}\}$}
		\ForEach{$\bm{H}_i \in \bm{\mathcal{H}}$}
		{
			\State{$\bm{H}_i =\sigma( \bm{W}^{l}\bm{H}_i) $ }
			\State{$\bm{\alpha}_{i} =\text{softmax}({a}(\bm{H}_0, \bm{H}_i))$ }
        }
        \State{$\bm{\mathcal{H}} = \text{AGG}(\bm{\alpha}_{i} \bm{H}_i), \text{ for } i = 0,...,k-1 $ }
        \State{$ \bm{\mathcal{H}}_{s} = \sigma(\bm{W}^r\bm{\mathcal{H}})$}
        \State{Initialize the subspace filter $\bm{S}$}        
        \State{$\bm{W_c'} = \text{ApplySparsity}(\bm{S}, \bm{W_c})$}
        \State{$\mathcal{\bm{O}} = \bm{\mathcal{H}}_{s} \times \bm{W_c'}$}
        \State{\text{return} $\mathcal{\bm{O}},  \bm{\mathcal{H}}_{s}, \bm{S}$}
	}
\end{algorithm}
Lines 1-2 represent the preprocessing stage, stacking aggregated features from different hops. 
Lines 3-8 describe the model propagation stage. 
The preprocessed features are linearly transformed by $\bm{W}^l$ and then fused by the weight from hop-wise attention. 
The fused hidden state is then transformed by $\bm{W}^r$. 
Lines 9-11 describe the subspace embeddings. 
This hop-wise attention mechanism enhances SMN's flexibility by attending to broader receptive fields, capturing the unique graph structure across different real-life datasets. 
\section{Online Search Phase}
This section provides the design details for extending the current \myblue{ML-based} model's online search phase~\cite{qdgnn,coclep} to OCS.
We then analyze the feasibility of applying the proposed method to OCIS to identify the intersection of multiple targets effectively. 


\subsection{Overlapping Communities Search (OCS)} \label{ocs}

Leveraging the subspace community embedding technique, we first extend the naive top-$k$ similarity search to an OCS named Sub-Topk.
Considering Sub-Topk's limitations, we proposed a spatial-aware algorithm called subspace cohesive community search (Sub-CS). 


\myparagraph{Sub-Topk} We first propose Lemma~\ref{lemma1} states that the classifier $\bm{W}_c'$ gated by SSF approximates community embeddings in overlapping community structures. 
Based on this result, we initiate a similarity-based approach called Sub-Topk to identify a query-dependent community. 
The algorithm takes the query nodes and the test set as input, mapping all nodes to the target subspace by performing an element-wised dot product against the basis vector.

\begin{lemma}[Sparse Classifier Approximates Global Centroid of Communities]\label{lemma1}
Given a set $X_j = \{x_1, x_2, ..., x_{|N_j|}\}$ of node embeddings in $\mathbb{R}^s$ that belong to community $j$, and a classifier vector $w_j \in \bm{W}_c'$. 
Through the learning process, $w_j$ will converges to the centroid $\mu_j$ of community $j$ defined by:
\[
\mu_j = \frac{1}{|N_j|} \sum_{x_i \in X_j} x_i.
\]
\end{lemma}

\begin{proof}
    The updated rule for using gradient descent is given by:
    \[
    w_j \leftarrow w_j + \alpha \sum_{i \in N_j} x_i \cdot (y_{ij} - \sigma(x_i^\top w_j)),
    \]
    where $y_{ij}$ is the indicator function, $\sigma$ is the activation function, and $\alpha$ is the learning rate.
    As the model learns, the predicted probabilities $\sigma(x_i^\top w_j)$ approach the true class labels $y_{ij}$, reducing the term $y_{ij} - \sigma(x_i^\top w_j)$ to a small error $\epsilon_{ij}$ near zero.
    \(
    \Delta w_j = \alpha \sum_{i \in N_j} x_i \epsilon_{ij},
    \)
    with $\epsilon_{ij}$ trending towards zero as classification accuracy improves.
    Assuming $\epsilon_{ij}$ becomes negligible, the updates to $w_j$ become smaller, stabilizing $w_j$ around a vector that maximizes the sum of projections of $x_i$ on $w_j$. This stabilization point is given by:
    \[
    \lim_{\epsilon_{ij} \to 0} w_j \approx \frac{\sum_{i \in N_j} x_i}{|N_j|} \approx \mu_j.
    \vspace{-3mm}
    \]
    \end{proof}
    \vspace{-2mm}
\myblue{\noindent A well-trained SSF functions as a subspace community embedding, effectively filtering out noisy nodes. 
This subspace mapping segregates the target community from overlapping communities.}

\begin{example}
For a node $v$ with a feature vector \( \mathbf{x} = [0.82, 0.11, -0.69, -1.3, 0.03] \), there are two distinct communities with basis vectors e.g., \( \bm{S}_1 = [1, 0, 1, 1, 0], \bm{S}_2 = [0, 1, 0, 0, 1]\). 
The projections are computed as follows: 
\[{x}_1' = {x} \odot \bm{S}_1 = [0.82, 0, -0.69, -1.3, 0],\]
\[{x}_2' = {x} \odot \bm{S}_2 = [0, 0.11, 0, 0, 0.03].\]
The vector ${x}_1'$ is likely to demonstrate higher similarity with nodes in the subspace defined by $\bm{S}_1$, compared to a lower probability with nodes in subspace $\bm{S}_2$. 
Therefore, the subspace mapping effectively filters out irreverent nodes during the search. 
\end{example}




\myparagraph{Sub-CS} While Sub-Topk effectively segregates unrelated nodes, it operates under an unrealistic assumption that the query node is always at the centroid of the community.
Inspired by the idea of spatial-aware community search~\cite{spatialcs_2015,spatialcs_2018}, we further introduce a subspace-aware community search (Sub-CS) to identify a densely interrelated community in the latent space, allowing a shift in the community centroid.   
The algorithm aims to locate a community with a small ``radius'' in the subspace. 
Subspace cohesiveness implies that the identified community should minimize the community radius in the latent subspace, with distance measured by cosine similarity. 
Sub-CS explores top-k nodes demonstrating the highest similarity to the query as the initial community; by traversing through the dataset following a descending order of node similarity, Sub-CS updates the community to maximize the group similarity. 
Let $v_q$ denote the query node, and $\mathbf{C_0}$ represent the initial community. 
Compute the centroid $\bar{x}$ of $\mathbf{C_0}$ as the mean of the nodes embeddings in $\mathbf{C_0}$, i.e., $\bar{x} = \frac{1}{|\mathbf{C}_0|} \sum_{v \in \mathbf{C}_0}x$, where $x$ is the embedding vector of node $v$.
Subsequently, we evaluate each node $v$ in $\mathbf{C}_0$ against $\bar{x}$ using cosine similarity. 
Replace the least similar node $v'$ with a new encountered node $v_i$ if $v_i$ exhibits higher similarity as \autoref{eq:comm}:
\begin{equation} \label{eq:comm}
\mathbf{C}_{i+1} = 
\begin{cases} 
\mathbf{C}_i \setminus \{v'\} \cup \{v_i\} & \text{if } \cos({x_i}, \bar{x}) > \cos({x'}, \bar{x}), \\
\mathbf{C}_i & \text{otherwise.}
\end{cases}
\end{equation}


This process continues and recomputes the centroid if the community is updated. 
The algorithm terminates if the query node becomes the least similar node or the early stop condition is reached.
The early stop controls the node similarity against the query, preventing the final community from including dissimilar nodes.
In the experiment, we set the threshold at $2 \times k$ nodes, which means that the algorithm will only consider nodes with top-2$k$ similarity. 
The details are illustrated in \autoref{alg:CS}. 
The Lemma \ref{lemma2} states that the community results in a decreasing radius by interactively removing the least similar node measured by cosine similarity.  
\begin{algorithm}[t]
	{
      \footnotesize

		\caption{Cohesive community search (Sub-CS)}\label{alg:CS}
		\Input{Graph $G$, Query ${v_q}$, final embeddings $\bm{\mathcal{H}}_{s}$, learned sparse subspace filter matrix $\bm{S}$,  community size ${k}$,  similarity threshold ${l}$}
		\Output{Community $\mathbf{C}_q$}
		\State{$\bm{\mathcal{H}}_{s} = \text{SubspaceMapping}(\bm{\mathcal{H}}_{s}, \bm{S}_q)$}
		\State{$\mathbf{C}_q = \{{q}\}$}
		\ForEach{${v}$ \text{encountered in} ${G}$ \text{sorted by similarity against} ${q}$}
		{
			\State{Add ${v}$ to $\mathbf{C}_q$ if $|\mathbf{C}_q| < {k} $ }
			\State{$\bar{x}= \text{mean}(\bm{\mathcal{H}}_{s}[i]), i \in \mathbf{C}_q$ }
   			\State{$\bm{P} = \text{sim}(\mathbf{C}_q, \bar{x})$}
   			\State{Find a node ${v'}\in \mathbf{C}_q$ with $\bm{P}[v']$ smallest in $\mathbf{C}_q$}
                \State{\textbf{if} ${v'} = {q}$ \textbf{then} break}
                \State{\textbf{if} $\text{sim}({v}, q) <= {l}$ \textbf{then} break}
                \State{$\bm{p}_v = \text{sim}({v}, \bar{x})$}
                \If{$\bm{p}_v > \bm{P}[v']$}
                    {
                    \State{$\mathbf{C}_q.\text{remove}({v'})$; $\mathbf{C}_q.\text{add}({v})$}
                    \State{$\bar{x} = \text{mean}(\bm{\mathcal{H}}_{s}[i]), i \in \mathbf{C}_q$ }
                    }       
        }
        \State{\text{return} $\mathbf{C}_q$}
	}
\end{algorithm}

\begin{lemma}[The Smallest Radius in Embedding Space]\label{lemma2}
    Let \( \mathbf{C}_{i}=\{x_1, x_2, \ldots, x_m\} \) be a set of points in \( \mathbb{R}^s \) is the embeddings of nodes in a identified community. Let \(\bar{x} = \frac{1}{m} \sum_{i=1}^m x_i\) be the centroid of these nodes. 
    Assume \(x'\in \mathbf{C}_{i}\) is the node with the minimum cosine similarity to \(\bar{x}\), and \(x_{m+1} \notin \mathbf{C}_{i}\), having \(\cos(\bar{x},x_{m+1}) > \cos(\bar{x},x')\). 
    When \(x'\) is replaced by \(x_{m+1}\), resulting in a new centroid \(\bar{x}'\), then:
    \[\sum_{i=1}^{m} \cos(\bar{x}', x_i) > \sum_{i=1}^{m} \cos(\bar{x}, x_i).\]
\end{lemma}



\begin{proof}
    Given \(\bar{x} = \frac{1}{m} \sum_{i=1}^m x_i\) is the centroid, the new centroid after replacing \(x'\) with \(x_{m+1}\) is:
    \(
    \bar{x}' = \bar{x} - \frac{1}{m} x' + \frac{1}{m} x_{m+1}.
    \)
    Define \(\bar{x}^* = \bar{x} - \frac{1}{m} x'\) as the centroid after removing \(x'\). Given \(\cos(\bar{x}, x_{m+1}) > \cos(\bar{x}, x')\), we have:
    \(
    \cos(\bar{x}^* + \frac{1}{m} x', x_{m+1}) > \cos(\bar{x}^* + \frac{1}{m} x', x').
    \)
    Since \(\cos(x', x') = 1\) and for any \(x_{m+1} \neq x'\), \(\cos(x', x_{m+1}) < 1\), it follows that:
    \(
    \cos(\bar{x}^*, x_{m+1}) > \cos(\bar{x}^*, x').
    \)
    As \(\bar{x}^*\) represents all other nodes in the community excluding \(x'\), the similarity of the new centroid \(\bar{x}'\) with each \(x_i\) increases:
    \[
    \cos(\bar{x}', x_i) = \frac{(\bar{x} - \frac{1}{m} x' + \frac{1}{m} x_{m+1}) \cdot x_i}{\|\bar{x} - \frac{1}{m} x' + \frac{1}{m} x_{m+1}\| \|x_i\|}.
    \]
    Summing these, \(\sum_{i=1}^m \cos(\bar{x}', x_i)\) is greater than \(\sum_{i=1}^m \cos(\bar{x}, x_i)\), proving Lemma \ref{lemma2}.
\end{proof}

\myblue{In OCS, cosine similarity, despite not being a metric space, offers key advantages for our task. 
Since nodes often belong to multiple communities, their feature vectors may have smaller values, and the magnitude can vary significantly between popular (multi-community) and less popular nodes. 
By emphasizing directional alignment rather than magnitude, cosine similarity allows us to focus on the structural similarity of nodes within a target community. 
This approach effectively handles variations in node popularity, enabling better identification of nodes aligned with the target community regardless of their overall influence in the network.}


\subsection{Overlapping Communities Intersection Search (OCIS)} \label{ocis}



In OCIS, SMN provides enhanced flexibility to end users by allowing the selection of multiple communities as the target and returning only to their intersection. 
The brute-force approach identifies all target communities and determines their intersection by examining common nodes. 
\myblue{However, this method leads to high computation overhead, as it requires enumerating the entire dataset for community prediction, followed by intersection-finding operations.
This process is equivalent to solving a community detection problem, which becomes time-consuming when handling large graphs, especially when the target result involves only a small subset of nodes.}
Leveraging subspace embedding techniques, SMN efficiently identifies the intersection while avoiding computational wastage. 
The rationale is that nodes in the intersection set should exhibit closer relationships with all community embeddings involved.
Lemma \ref{lemma3} states that if nodes demonstrate high similarity in two subspaces, they will also be similar in their unioned space.  

\begin{lemma}[Cosine Similarity Preserved in Unioned Subspace]\label{lemma3}
Given \( k \) communities, each represented by a distinct subspace \( \bm{S}_i \) where \( i \in \{1, 2, \ldots, k\} \). Define the unioned subspace \( U \) as \( U = \bigcup_{i=1}^{k} \bm{S}_i \). If two nodes \( x_1 \) and \( x_2 \) demonstrate high cosine similarity in each distinct subspace \( \bm{S}_i \), then \( x_1 \) and \( x_2 \) will also demonstrate substantial cosine similarity in the unioned subspace \( U \).
\end{lemma}
\begin{proof}
Let \( x_1^{(i)} \) and \( x_2^{(i)} \) be the representations of \( x_1 \) and \( x_2 \) in the subspace \( \bm{S}_i \). Since \( x_1 \) and \( x_2 \) demonstrate high cosine similarity in subspaces \( \bm{S}_i\), we have:
\(
\cos(x_1^{(i)}, x_2^{(i)}) = \frac{x_1^{(i)} \cdot x_2^{(i)}}{\|x_1^{(i)}\| \|x_2^{(i)}\|} \approx 1 \quad \forall i \in \{1, 2, \ldots, k\}
\).
Since \( x_1^U \cdot x_2^U \) is the sum of the dot products in each subspace \( \bm{S}_i \):
\(
x_1^U \cdot x_2^U = \sum_{i=1}^k x_1^{(i)} \cdot x_2^{(i)}
\).
Given that \( \cos(x_1^{(i)}, x_2^{(i)}) \approx 1 \), we have:
\(
x_1^{(i)} \cdot x_2^{(i)} \approx \|x_1^{(i)}\| \|x_2^{(i)}\| \quad \forall i \in \{1, 2, \ldots, k\}
\).
Therefore:
\[
x_1^U \cdot x_2^U \approx \sum_{i=1}^k \|x_1^{(i)}\| \|x_2^{(i)}\| \approx \|x_1^U\| \|x_2^U\|.
\]
This shows that if \( x_1 \) and \( x_2 \) demonstrate high cosine similarity in each distinct subspace \( \bm{S}_i\), they will also demonstrate substantial high cosine similarity in the unioned subspace \( U\).
\end{proof}

\section{Theoretical Analysis}\label{train & query}
\myparagraph{Analysis of the query-graph encoder}\label{theoretical}
We first present a theoretical analysis explaining why the widely used query-graph encoder frameworks in \myblue{ML-based} approaches~\cite{qdgnn, coclep} cause computational overhead leading to $O(|V|^2)$ training time complexity.

\begin{lemma}[Time Complexity of the Query-Graph Encoder Framework]\label{lemma4}
    Fusing the hidden states of both encoders' output at each layer will lead to $O(|V|^2)$ time complexity in model training.
\end{lemma}
\begin{proof}
    Given that the query encoder processes each node individually, similar to Stochastic Gradient Descent (SGD), the training will involve $|V|$ batches, each sized as 1. 
    The time complexity is $O(1)$ for each batch, and for the entire training set, it is $|V| \times O(1) = O(|V|)$.
    The graph encoder processes the full graph each time with a single batch sized $|V|$.
    The time complexity is $O(|V|)$ for each batch, and for the entire training set, it remains $O(|V|) = 1 \times O(|V|)$.
    When the two encoders are fused, there are $|V|$ batches, each sized as $(1, |V|)$.
    This means that for each node processed by the query encoder, the graph encoder processes the entire graph. 
    Consequently, the total time complexity for each batch increases to $|V|$, which results in $|V| \times O(|V|) = O(|V|^2)$.  
\end{proof}

Hence, the query-graph encoder framework is the primary reason for the slow training issues in existing CS models. 
Moreover, while this framework is claimed to capture both local and global information, it primarily affects the gradient descent optimization. 
Ultimately, this impact averages out over the training, behaving similarly to standard SGD as shown in Lemma~\ref{lemma5}.

\begin{lemma}[The optimization is Equivalence to SGD]\label{lemma5}
Let $\theta$ be the parameter vector of a neural network trained using a query-graph encoder where $Q(v_i)$ adjusts $\theta$ locally for each node $v_i$ and $G(G)$ provides global adjustments based on the entire graph $G$. The cumulative effect of these adjustments over multiple training epochs is equivalent to the effect of Stochastic Gradient Descent (SGD) on $\theta$.
\end{lemma}
\begin{proof}
    The local adjustments by $Q(v_i)$ for each node and the global adjustments by $G(G)$ can be formally expressed as:
    \[
    \theta \leftarrow \theta - \eta \left( \nabla L_{\text{local}}(Q(v_i), \theta) + \nabla L_{\text{global}}(G(G), \theta) \right),
    \]
    where $\eta$ is the learning rate, and $\nabla L_{\text{local}}$ and $\nabla L_{\text{global}}$ are the gradients of the loss functions localized to $Q(v_i)$ and globalized to $G(G)$, respectively.
    In traditional SGD, parameter updates are influenced by the gradient of the loss function evaluated at different subsets of the data. Over many iterations, this results in:
    \[
    \theta \leftarrow \theta - \eta \frac{1}{n} \sum_{i=1}^n \nabla L_{\text{local}}(Q(v_i), \theta),
    \]
    where $n$ is the total number of nodes.
    Given the high frequency of updates involving every node $v_i$ and the entire graph $G$, the effects of $Q(v_i)$ and $G(G)$ fusion leading  each parameter update by $Q(v_i)$ is averaged with updates induced by $G(G)$:
    \[
    \theta \leftarrow \theta - \eta \left( \frac{1}{n} \sum_{i=1}^n \nabla L_{\text{local}}(Q(v_i), \theta) + \nabla L_{\text{global}}(G(G), \theta) \right).
    \]
    This aligns with the principle of SGD, which states that the aggregate update is the average of the updates across all data points.
\end{proof}


\myparagraph{SMN time complexity analysis} 
We provide detailed considerations for both model training and community identification, addressing preprocessing and query time complexities, respectively. 
The feature processing adopts $k$-hop operations in the preprocessing stage, contributing $O(|V|^{3} \times k)$.
The subsequent multi-hop processing layer involves $O(|V| \times d \times h)$, where $d$ is the initial feature dimensions, and $h$ is the hidden dimensions.
Notably, for vanilla GNN models, the aggregation process with $O(|V|^{3} \times k)$ happens during the model training. 
This slows down the model training speed due to repetitive propagation and backpropagation operations.
Comparatively, the above preprocessing only operates once before training, avoiding expensive overhead during the training. 
The multi-hop attention introduces $O(|V| \times i \times h)$ complexity, where $i$ is the number of heads. Hop-wise addition and weighted average fusion will take $O(|V|)$. 
As these operations will be run for $k$ time, the total time complexity for the multi-hop processing and multi-hop attention layers is $O(k \times |V| \times h \times (d + i))$. 
The multi-head filter and SSF transformation will take $O(|V| \times h \times s)$ and $O(|V| \times s \times c)$. 
Where $s$ represents the dimensions of SSF and $c$ is the number of communities. The model is trained by t epochs. 
Therefore, the total time complexity for SMN training is $O(|V|^{3} \times k + (|V|\times t \times (h \times k \times (d + i) + s \times (h + c))))$. 

For Sub-Topk, applying the target SSF to map node features and computing cosine similarity against the query node will take $O(|V| \times h)$. 
To get the top $k$ similarity, will take $O(|V| \times log(k))$, where $k$ is the community size. 
Therefore, the total time complexity for Sub-Topk will be $O(|V| \times (h+log(k)))$.

For Sub-CS, applying the target SSF to map node features will take $O(|V| \times h)$. 
To get the top $k$ similarity, will take $O(|V| \times log(k))$, where $k$ is the community size. 
To update the centroid will take $O(k^2 \times h)$. 
Therefore, the total time complexity will be $O((k^2 \times h) + |V| \times (h+log(k)))$.
\section{Experiments} \label{Experiments}
In this section, we conduct experiments on 13 datasets to demonstrate the effectiveness and efficiency of the proposed techniques from 5 perspectives.
We first report the model performance on OCS and OCIS to demonstrate the effectiveness in overlapping community structures. 
Then, we illustrate the efficiency comparison to show the superiority of SMN in both model training and query processing.
Thirdly, we compare the model performance in disjoint datasets to show that even though the SMN and SSF are primarily built for OCS, they can also effectively handle disjoint community search.
Subsequently, we conduct the ablation study to analyze the contribution of each building block. 
Finally, we provide hyper-parameters analysis to reveal insights into model parameters.

\begin{table}[t]
  \caption{\myblue{Dataset statistics}}
  \label{tab:data}
  \vspace{-2mm}
  \footnotesize
  \begin{tabular}{ c|c||c|c|c|c|c|c }
    \toprule
     \multicolumn{2}{c||}{Dataset} & \multicolumn{1}{c}{\text{\# Nodes}} & \multicolumn{1}{|c}{\text{\# Edges}} & \multicolumn{1}{|c}{\text{\# Com}} & \multicolumn{1}{|c}{\text{\# Feat}} & \multicolumn{1}{|c}{\myblue{\text{OR}}} & \multicolumn{1}{|c}{\myblue{\text{MLA}}} \\
    \midrule
    \multirow{9}{*}{Overlap} & FB-0 & 185 & 645 & 3 & 224 & 0.188 & 3 \\
    & FB-107 & 418 & 4,815 & 4 & 576 & 0.02 & 2 \\
    & FB-348 & 207 & 2,716 & 4 & 161 & 0.744 & 4 \\
    & FB-414 & 108 & 954 & 2 & 105 & 0.065 & 2 \\
    & FB-686 & 159 & 1,607 & 6 & 63 & 0.698 & 6 \\
    & Chemistry & 35,409 & 157,358 & 14 & 4,877 & 0.25 & 13 \\
    & CS & 21,957 & 96,750 & 18 & 7,793 & 0.275 & 13 \\
    & Engineering & 14,927 & 49,305 & 16 & 4,839 & 0.272 & 12 \\
    & Medicine & 63,282 & 810,314 & 17 & 5,538 & 0.365 & 16 \\
    \midrule
    \multirow{4}{*}{Disjoint} & Cora & 2,708 & 5,429 & 7 & 1,433 & - & - \\
    & Citeseer & 3,312 & 4,732 & 6 & 3,703 & - & - \\
    & Pubmed & 19,717 & 44,338 & 3 & 500 & - & - \\
    & Reddit & 232,965 & 114M & 41 & 602 & - & - \\
    \bottomrule
  \end{tabular}
\end{table}



\begin{table*}[t]
  \vspace{-3mm}
  \centering
  \caption{\myblue{SMN performance in overlapping community search}}
  \vspace{-3mm}
    \begin{adjustbox}{width=\linewidth}
    \label{tab:performance_overlap}
    \begin{tabular}{c|c||cccccccc|cccccccc|c}
    \toprule
        ~ & \textbf{Task} & \multicolumn{8}{c|}{\textbf{Overlapping Community Search, OCS}} &\multicolumn{8}{c|}{\textbf{Overlapping Communities Intersection Search, OCIS}} &   \\ 
        Metric & Model & \makecell{k-clque} & \makecell{CTC} & \makecell{k-core} & \makecell{ICS \\ GNN} & \makecell{QD \\ GNN} & \makecell{COC \\ LEP} & \makecell{SMN \\ Topk} & \makecell{SMN \\ CS} & \makecell{k-clque} & \makecell{CTC} & \makecell{k-core} & \makecell{ICS \\ GNN} & \makecell{QD \\ GNN} & \makecell{COC \\ LEP} & \makecell{SMN \\ Topk} & \makecell{SMN \\ CS} & \makecell{\myblue{Ave+}}  \\ 
            \midrule
            \midrule
        \multirow{9}{*}{F1} &  FB0 & 0.2478 & 0.2588 & 0.2423 & 0.7058 & 0.6710 & 0.2424 & \underline{0.7427} & \textbf{0.7630} & 0.0572 & 0.0622 & 0.0551 & 0.6122 & 0.5982 & 0.6667 & \underline{0.6547} & \textbf{0.7147} & \myblue{35\%}  \\ 
         &  FB107 & 0.2781 & 0.3024 & 0.2537 & 0.6835 & 0.6361 & - & \underline{0.9035} & \textbf{0.9103} & 0.0712 & 0.0829 & 0.0609 & 0.5127 & 0.5760 & - & \textbf{0.7520} & \underline{0.6520} & \myblue{46\%}  \\ 
         &  FB348 & 0.1543 & 0.1366 & 0.1443 & \underline{0.8041} & 0.7338 & 0.6907 & \textbf{0.8517} & 0.7913 & 0.0916 & 0.0949 & 0.0840 & 0.7508 & 0.7316 & 0.6822 & \textbf{0.8114} & \underline{0.8031} & \myblue{39\%}  \\ 
         &  FB414 & 0.2882 & 0.3119 & 0.2718 & 0.7941 & 0.6923 & 0.7286 & \underline{0.8745} & \textbf{0.9006} & 0.0798 & 0.0907 & 0.0681 & 0.4107 & 0.4813 & 0.2080 & \underline{0.7493} & \textbf{0.7533} & \myblue{45\%}  \\ 
         &  FB686 & 0.0947 & 0.0881 & 0.1013 & 0.6366 & 0.6006 & 0.6512 & \underline{0.6776} & \textbf{0.7075} & 0.0691 & 0.0825 & 0.0615 & 0.4077 & 0.4351 & 0.4201 & \underline{0.4958} & \textbf{0.5966} & \myblue{32\%}  \\ 
         & ENG & 0.0471 & 0.0529 & 0.0553 & 0.6680 & 0.7422 & 0.1530 & \textbf{0.8172} & \underline{0.7618} & 0.0471 & 0.0529 & 0.0553 & 0.6406 & 0.6792 & 0.1659 & \textbf{0.8096} & \underline{0.7973} & \myblue{52\%}  \\ 
         & CS & 0.0395 & 0.0433 & 0.0408 & 0.6187 & 0.5878 & 0.1400 & \textbf{0.8301} & \underline{0.8242} & 0.0395 & 0.0433 & 0.0408 & 0.6426 & 0.6472 & 0.1507 & \underline{0.7383} & \textbf{0.7504} & \myblue{53\% } \\ 
         & CHEM & 0.0594 & 0.0615 & 0.0623 & 0.5732 & 0.6151 & 0.1812 & \textbf{0.8585} & \underline{0.8487} & 0.0594 & 0.0615 & 0.0623 & 0.6047 & 0.6940 & 0.2199 & \underline{0.8734} & \textbf{0.8758} & \myblue{59\%}  \\ 
         & MED & - & 0.0503 & 0.0622 & 0.6630 & 0.5704 & 0.1628 & \underline{0.8416} & \textbf{0.8540} & - & 0.0503 & 0.0622 & 0.6760 & 0.6927 & 0.1651 & \underline{0.8405} & \textbf{0.8514} & \myblue{53\%}  \\ 
        \midrule
        \multirow{9}{*}{JAC} &  FB0 & 0.1972 & 0.2115 & 0.1903 & 0.5446 & 0.5049 & 0.1379 & \underline{0.5907} & \textbf{0.6168} & 0.0559 & 0.0609 & 0.0538 & 0.6022 & 0.5811 & 0.5172 & \underline{0.6500} & \textbf{0.7080} & \myblue{34\%}  \\ 
         &  FB107 & 0.2386 & 0.2768 & 0.2048 & 0.5192 & 0.4664 & - & \underline{0.8783} & \textbf{0.8913} & 0.0709 & 0.0827 & 0.0606 & 0.5113 & 0.5760 & - & \textbf{0.7520} & \underline{0.6520} & \myblue{49\%}  \\ 
         &  FB348 & 0.1116 & 0.0940 & 0.1128 & \underline{0.6724} & 0.5796 & 0.5275 & \textbf{0.7417} & 0.6547 & 0.0874 & 0.0924 & 0.0771 & 0.6649 & 0.6447 & 0.5460 & \textbf{0.7233} & \underline{0.7157} & \myblue{36\%}  \\ 
         &  FB414 & 0.2538 & 0.2931 & 0.2294 & 0.6585 & 0.5294 & 0.5731 & \underline{0.7769} & \textbf{0.8191} & 0.0795 & 0.0903 & 0.0673 & 0.3987 & 0.4680 & 0.2080 & \underline{0.7380} & \textbf{0.7420} & \myblue{45\%}  \\ 
         &  FB686 & 0.0661 & 0.0599 & 0.0728 & 0.4669 & 0.4292 & 0.4828 & \underline{0.5124} & \textbf{0.5474} & 0.0641 & 0.0796 & 0.0554 & 0.3623 & 0.4002 & 0.2793 & \underline{0.4645} & \textbf{0.5597} & \myblue{29\%}  \\ 
         & ENG & 0.0260 & 0.0296 & 0.0311 & 0.5015 & 0.5901 & 0.0828 & \textbf{0.6908} & \underline{0.6152} & 0.0260 & 0.0296 & 0.0311 & 0.6259 & 0.6634 & 0.0917 & \textbf{0.7799} & \underline{0.7659} & \myblue{49\%}  \\ 
         & CS & 0.0224 & 0.0249 & 0.0233 & 0.4479 & 0.4162 & 0.0753 & \textbf{0.7096} & \underline{0.7009} & 0.0224 & 0.0249 & 0.0233 & 0.6124 & 0.6244 & 0.0839 & \underline{0.7101} & \textbf{0.7206} & \myblue{51\%}  \\ 
         & CHEM & 0.0349 & 0.0363 & 0.0369 & 0.4017 & 0.4442 & 0.0996 & \textbf{0.7522} & \underline{0.7372} & 0.0349 & 0.0363 & 0.0369 & 0.5744 & 0.6728 & 0.1298 & \textbf{0.8403} & \underline{0.8392} & \myblue{58\%}  \\ 
         & MED & - & 0.0288 & 0.0368 & 0.4959 & 0.3990 & 0.0886 & \underline{0.7266} & \textbf{0.7453} & - & 0.0288 & 0.0368 & 0.6404 & 0.6472 & 0.0933 & \underline{0.7946} & \textbf{0.8054} & \myblue{52\%}  \\ 
        \midrule
        
        \multirow{9}{*}{{\myblue{NMI}}} &  \myblue{FB0} & \myblue{0.1788} & \myblue{0.1245} & \myblue{0.2069} & \myblue{0.1535} & \myblue{0.2007} & \myblue{0.1029} & \textbf{\myblue{0.3182}} & \underline{\myblue{0.2905}} & \myblue{0.1788} & \myblue{0.1245} & \myblue{0.2069} & \myblue{0.2117} & \myblue{0.2021} & \myblue{0.1673} & \underline{\myblue{0.5212}} & \textbf{\myblue{0.5418}} & \myblue{25\%}  \\ 
         &  \myblue{FB107} & \myblue{0.3790} & \myblue{0.5479} & \myblue{0.2054} & \myblue{0.1590} & \myblue{0.2794} & - & \textbf{\myblue{0.6176}} & \underline{\myblue{0.5937}} & \myblue{0.3790} & \myblue{0.5479} & \myblue{0.2054} & \myblue{0.1554} & \myblue{0.2043} & - & \textbf{\myblue{0.6395}} & \underline{\myblue{0.6197}} & \myblue{31\%}  \\ 
         &  \myblue{FB348} & \myblue{0.3338} & \myblue{0.4700} & \myblue{0.3321} & \myblue{0.4626} & \myblue{0.4155} & \myblue{0.2345} & \underline{\myblue{0.5301}} & \textbf{\myblue{0.6550}} & \myblue{0.3338} & \myblue{0.3380} & \myblue{0.3321} & \myblue{0.2023} & \myblue{0.1760} & \myblue{0.0771} & \textbf{\myblue{0.3829}} & \underline{\myblue{0.3582}} & \myblue{17\%}  \\ 
         &  \myblue{FB414} & \myblue{0.3695} & \myblue{0.4281} & \myblue{0.4250} & \myblue{0.4449} & \myblue{0.3914} & \myblue{0.3189} & \underline{\myblue{0.5669}} & \textbf{\myblue{0.6186}} & \myblue{0.3695} & \myblue{0.4281} & \myblue{0.4250} & \myblue{0.3529} & \myblue{0.4286} & \myblue{0.1375} & \underline{\myblue{0.6318}} & \textbf{\myblue{0.6325}} & \myblue{24\%}  \\
         &  \myblue{FB686} & \myblue{0.2862} & \myblue{0.2790} & \myblue{0.2225} & \myblue{0.2047} & \myblue{0.2864} & \myblue{0.1773} & \textbf{\myblue{0.4040}} & \underline{\myblue{0.3777}} & \myblue{0.2862} & \myblue{0.2790} & \myblue{0.2225} & \myblue{0.2474} & \myblue{0.2662} & \myblue{0.2608} & \textbf{\myblue{0.4723}} & \underline{\myblue{0.4495}} & \myblue{17\%}   \\ 
         &  \myblue{ENG} & \myblue{0.0424} & \myblue{0.0545} & \myblue{0.0687} & \myblue{0.3201} & \myblue{0.4550} & \myblue{0.0325} & \textbf{\myblue{0.5803}} & \underline{\myblue{0.4810}} & \myblue{0.0424} & \myblue{0.0545} & \myblue{0.0687} & \myblue{0.3094} & \myblue{0.4986} & \myblue{0.0333} & \textbf{\myblue{0.7696}} & \underline{\myblue{0.7590}} & \myblue{48\%}  \\ 
         &  \myblue{CS} & - & \myblue{0.0377} & - & \myblue{0.2936} & \myblue{0.2983} & \myblue{0.0097} & \underline{\myblue{0.5954}} & \textbf{\myblue{0.6033}} & - & \myblue{0.0377} & - & \myblue{0.2985} & \myblue{0.4734} & \myblue{0.0047} & \underline{\myblue{0.6891}} & \textbf{\myblue{0.7043}} & \myblue{47\% } \\ 
         & \myblue{CHEM} & \myblue{0.0393} & \myblue{0.0396} & \myblue{0.0411} & \myblue{0.2745} & \myblue{0.2961} & \myblue{0.0297} & \textbf{\myblue{0.6546}} & \underline{\myblue{0.6405}} & \myblue{0.0393} & \myblue{0.0396} & \myblue{0.0411} & \myblue{0.2636} & \myblue{0.4930} & \myblue{0.0107} & \textbf{\myblue{0.7028}} & \underline{\myblue{0.6896}} & \myblue{54\% }  \\ 
         & \myblue{MED} & - & \myblue{0.0556} & \myblue{0.0430} & \myblue{0.3916} & \myblue{0.2744} & \myblue{0.0746} & \underline{\myblue{0.6419}} & \textbf{\myblue{0.6726}} & - & \myblue{0.0556} & \myblue{0.0430} & \myblue{0.3806} & \myblue{0.4606} & \myblue{0.0303} & \underline{\myblue{0.6728}} & \textbf{\myblue{0.6876}} & \myblue{49\%} \\ 

    \bottomrule
    \end{tabular}
    \end{adjustbox}
\end{table*}

\begin{figure*}[t]
\vspace{-7mm}
    \centering
    \begin{subfigure}[b][4cm][c]{0.49\linewidth}
        \centering
        \includegraphics[width=\linewidth]{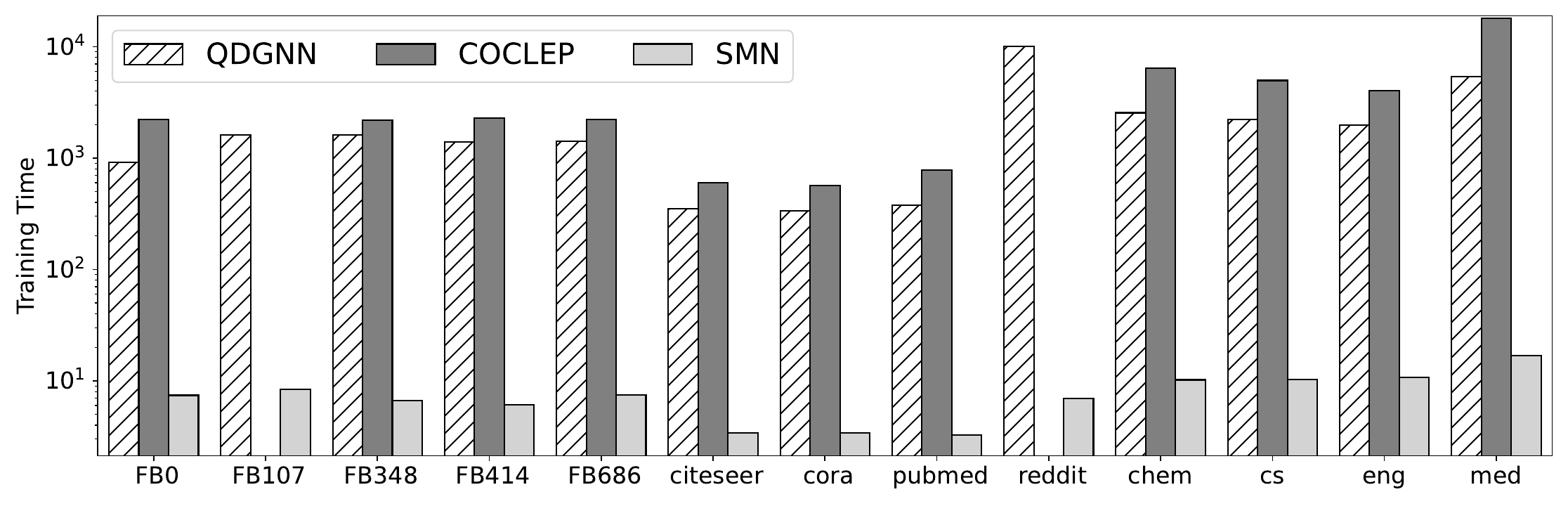}
        \vspace{-5mm}
        \caption{Efficiency results of the training phase}
        \label{fig:efficiency_1}
    \end{subfigure}
    \begin{subfigure}[b][4cm][c]{0.49\linewidth}
        \centering
        \includegraphics[width=\linewidth]{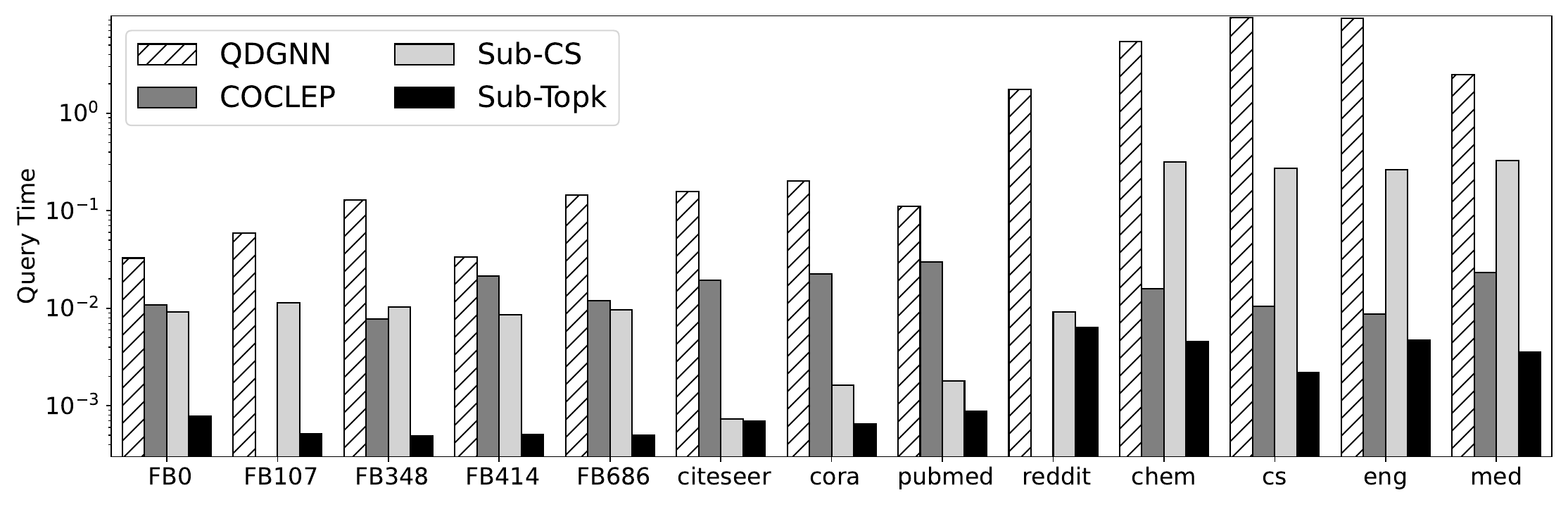}
        \vspace{-5mm}
        \caption{Efficiency results of the query phase}
        \label{fig:efficiency_2}
    \end{subfigure}
        \vspace{-7mm}
    \caption{Efficiency evaluation of different datasets (in seconds)}
        \vspace{-4mm}
\label{fig:efficiency}        
\end{figure*}
\subsection{Experimental Setup}


\myparagraph{Datasets} We use $13$ datasets to evaluate the performance of SMN, including $9$ datasets with overlapping community structures and $4$ datasets demonstrating disjoint structures. 
Datasets statistics are reported in \autoref{tab:data}. 
Facebook~\cite{Facebook} is a social network dataset that contains five ego networks. 
Chemistry, Computer Science, Medicine, and Engineering are co-authorship networks constructed using data from the Microsoft Academic Graph (MAG) \footnote{https://www.microsoft.com/en-us/research/project/open-academic-graph/}.
Cora, Citeseer, and Pubmed are citation networks with details disclosed on Relational Dataset Repository\footnote{https://relational.fit.cvut.cz/}. 
Reddit~\cite{Reddit} is an online forum where nodes are posts, and edges are comments from the same user. 
\myblue{Overlap Ratio ($OR = \frac{1}{n} \sum_{i=1}^{n} \mathbb{I}(|L_i| > 1)$) and Max Label Affiliations ($MLA = \max\left(|L_i|\right)$) measure how overlap the dataset is.}

\myparagraph{Data splitting} \myblue{By following the popular semi-supervised settings~\cite{sgc_2019}, SMN uses $10\%$ or less labeled data during the training to mitigate the human effort on labeling. 
In citation networks, a standard splitting is applied~\cite{lightgcn_2020}, utilizing only 20 samples from each community, accounting for less than $2\%$ of the total data.}
For the Facebook, Reddit, and MAG datasets, the splitting ratio is 10:10:80.
\myblue{The training set is exclusively used to compute the loss and update model parameters. 
During validation, the parameters are frozen, and the test set remains untouched to prevent information leakage.}

\myparagraph{Baseline models} We compare the performance of SMN against three algorithm-based methods ($k$-clique~\cite{k-clique}, CTC~\cite{k-truss2015}, and $k$-core based~\cite{k-core}) and three SOTA GNN-based models (ICS-GNN~\cite{ics-gnn}, QDGNN~\cite{qdgnn}, and COCLEP~\cite{coclep}). 
GNN-based models are all primarily focused on disjoint community structure. 
ICS-GNN shares the same configuration as the proposed SMN, identifying a $k$-sized community. 
QDGNN and COCLEP identify communities by training a threshold to measure GNN score and similarity. 
We extend their configuration to a $k$-sized community search by selecting the top-$k$ nodes with the highest GNN score or similarity.
The hyper-parameters settings are the same as in their original paper.

\myparagraph{Query setting} 
In the experiment, all query nodes are randomly selected to prevent potential bias. 
The community size $k$ is dependent on users, which can differ among datasets. 
In OCS, we set 30 as the community size for Facebook, 150 for Cora, Citeseer and Pubmed, and 1000 for MAG and Reddit.
In OCIS, as the intersection set is smaller in size, we use $k/5$ as the community size, where $k$ is the community size of the underlying dataset in OCS.

\myparagraph{Evaluation metrics} \myblue{The evaluation of identified communities is conducted through two performance metrics: F1-Score~\cite{ics-gnn,qdgnn}, Jaccard similarity (JAC)~\cite{coclep}, and overlapping NMI~\cite{shchur2019overlapping_NMI}.
The F1-Score balances precision and recall, offering a measure of how well the identified community matches the ground truth. 
JAC evaluates the overlap between the predicted and true communities by comparing the intersection over the union of the two sets. 
NMI focuses on the alignment between predicted and true overlapping communities, capturing the amount of shared structural information between them.}
The true data is established as the target community label, with the labels of the identified nodes serving as the predicted data. 
To evaluate the efficiency, the model training time and online querying time are recorded across different models. 
All results are averaged across 50 randomly selected queries to ensure the quality of the evaluation process. 

\myparagraph{Implementation details} SMN is constructed with 16-hop receptive fields, 128 hidden dimensions, and two heads for multi-head attention, using a 64-dimensional SSF. 
Model training involves a learning rate of 0.02 with 100 and 300 epochs for disjoint and overlapping datasets. 
Due to its size, SMN is configured with 4-hop receptive fields for MAG and Reddit, while the other hyper-parameters remain the same.
The experiments are run on a machine with Intel Xeon  6248R CPU, Nvidia A5000 GPU, and 512GB memory.
The code is available at anonymous Github \footnote{https://anonymous.4open.science/r/SMN-86B4/}.

\subsection{Overlapping Community Search}

\begin{table}[t]
  \centering
  \begin{adjustbox}{width=\linewidth}
    \begin{threeparttable}
    \caption{SMN performance (F1-Score) in Disjoint community search}
    \vspace{-2mm}
    \label{tab:disjoint}  
    \begin{tabular}{c||c|c|c|c|c|c|c}
    \toprule
    Datasets & $k$-clque & CTC & $k$-core & ICSGNN & QDGNN & COCLEP & SMN  \\ 
    \midrule
        Cora & 0.2941 & 0.3179 & 0.309 & 0.7787 & 0.7208 & 0.2516 & \textbf{0.8866}  \\ 
        Citeseer & 0.2951 & 0.281 & 0.3081 & 0.7679 & 0.7486 & 0.3944 & \textbf{0.7698}  \\ 
        Pubmed & 0.5121 & 0.5427 & 0.5586 & 0.8065 & 0.7999 & 0.5153 & \textbf{0.8255}  \\ 
        Reddit & - & 0.1271 & 0.2681 & 0.7374 & 0.8273 & - & \textbf{0.9433}  \\ 
    \bottomrule
    \end{tabular}
    \end{threeparttable}
  \end{adjustbox}
  \vspace{-2mm}
\end{table}

\myparagraph{Model effectiveness} \autoref{tab:performance_overlap} illustrates the model performance on overlapping community datasets. 
The comparison is between baseline models, SMN (Sub-Topk), and SMN (Sub-CS) on the Facebook and MAG datasets.
Missing results in $k$-clique and COCLEP are either caused by out-of-memory or failure to assign nodes to any community affiliation.
\myblue{The $Ave+$ represents the average performance improvement compared to the proposed SMN against all baseline models.}
In OCS, instead of selectively choosing a single community as the target, we use each community from the predicted list as the target and report the average performance to avoid bias. 
In OCIS, we use the full label list of the query as the target set and search for nodes that fall in at least all target communities, where more affiliations are not punished. 
Among the 9 datasets, Sub-CS demonstrates the best performance on 5 and 6 datasets for OCS and OCIC, respectively. Sub-Topk achieves the best results for others. 
\myblue{Notably, the algorithm-based models generate comparatively low performance in F1, JAC, and NMI, respectively.
This is mainly caused by the fact that those models are not task-driven and fail to predict the ground truth label.}

\myblue{In the OCS task, the ML-based models exhibit lower performance on the MAG datasets compared to the Facebook datasets.
This discrepancy is primarily due to the size and complexity of the datasets. 
As shown in \autoref{tab:data}, the MAG datasets have relatively lower OR and higher MLA, indicating a higher variance in label distribution.
For instance, in the Chemistry dataset, 75\% of nodes belong to only one community, while the most popular nodes are associated with up to 13 communities.
Despite this significant label variance, the proposed SMN maintained stable and superior performance, achieving $Ave+$ improvement of 59\%, 58\%, and 54\% in F1, JAC, and NMI, respectively.
These results demonstrate the model's robustness and effectiveness in handling overlapping community structures. 
}

\myblue{

When extending to OCIS, we used the query’s full label set as the target community set. In datasets like Chemistry, popular query nodes may be linked to as many as 13 communities, creating a unique challenge that can lead to out-of-sample issues. 
These issues may hinder the model's ability to identify $k$ nodes that all meet the specific criteria. 
As shown in \autoref{tab:data}, baseline models experience a significant performance drop when applied to OCIS, whereas SMN demonstrates relatively stable performance across most datasets, with only a slight decline in F1-Score compared to its performance in OCS. 
This highlights SMN's superiority in handling these more complex scenarios.}

\myblue{
F1 and JAC focus on the exact matching of each node's prediction, while NMI emphasizes the similarity between the predicted community and the ground truth community. 
Despite these differences, the experiment shows similar trends across all three metrics.
On average, our proposed SMN model surpasses the best baseline in OCS by 13.50\% and in OCIS by 13.96\% based on F1-Score, and by 19.19\% and 19.95\% respectively in terms of NMI.
This highlights the effectiveness of SMN in handling diverse community search tasks.}

\myparagraph{Model training efficiency} Figure~\ref{fig:efficiency}(a) depicts model training time comparison.
The reported training time for SMN includes times for both preprocessing and model training.
As ICS-GNN operates in an online setting and is trained and tested on candidate subgraphs, its training time is not included.
Notably, QDGNN and COCLEP exhibit significantly longer training time. 
\myblue{When training for the small Facebook datasets by 300 epochs, QDGNN and COCLEP take over 1,000 seconds, whereas SMN requires less than 10 seconds. 
On the large and densely connected dataset (Reddit), training SMN for 100 epochs takes only 7.6 seconds, compared to QDGNN's 11,055.1 seconds, and COCLEP runs out of memory (OOM).}
These results empirically demonstrate the efficiency of the proposed lightweight SMN and support Lemma~\ref{lemma4}.  
The proposed model framework accelerates model training by 2 orders of magnitude on average and achieves up to 3 orders of magnitude on large graphs such as Reddit and MAG-Medicine.

\myparagraph{Online query efficiency} Figure~\ref{fig:efficiency}(b) shows online query performance. 
QDGNN exhibits slow querying speed due to its reliance on BFS as the backbone algorithm.
Overall, Sub-Topk achieves the best results, surpassing the existing best one by 2 orders.
Sub-CS demonstrates superior efficiency on smaller datasets compared to COCLEP but slows down on larger datasets due to the necessity of recomputing the centroid each time the community is updated, making it sensitive to community size.
The trade-off between Sub-Topk and Sub-CS suggests that Sub-CS is preferable due to its accuracy for complex tasks like OCIS, where the required community size is small. 
However, for the tasks favoring efficiency with large community sizes, Sub-Topk will be preferred.

\begin{figure}[t]
  \vspace{-5mm}
  \centering
  \includegraphics[width=0.5\textwidth]{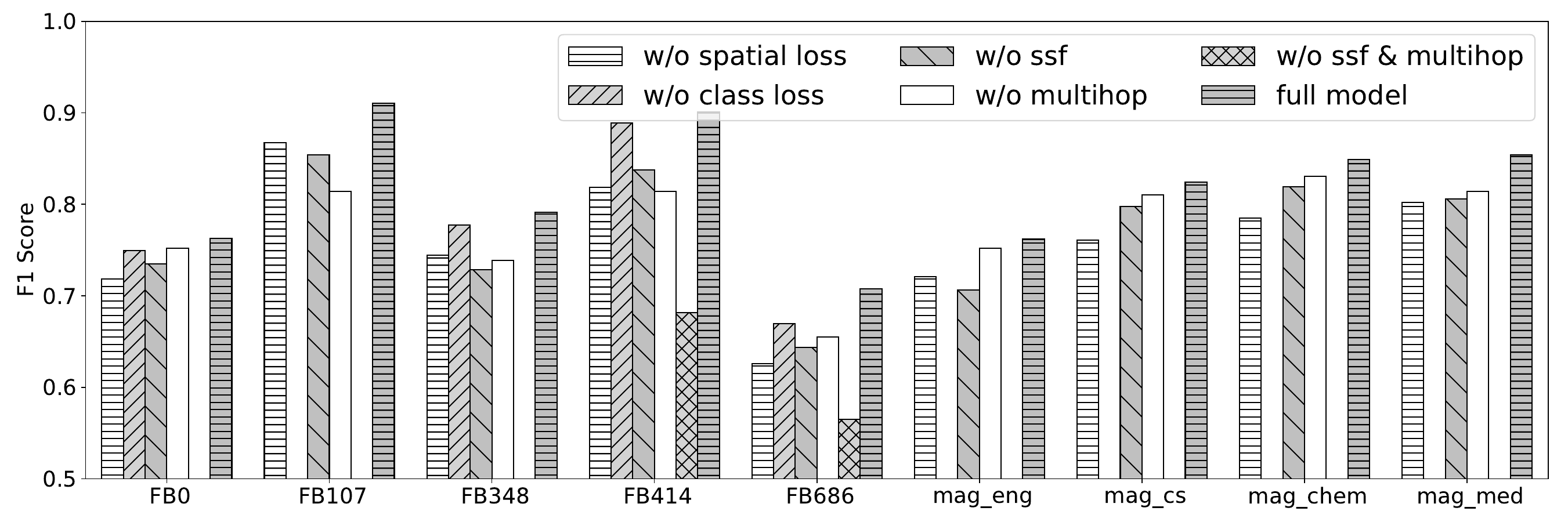}
\vspace{-8mm}
  \caption{Ablation study}
  \label{fig:ablation}
\vspace{-2mm}
\end{figure}

\begin{figure*}[t]
\vspace{-8mm}
    \centering
    \begin{subfigure}[b][4cm][c]{0.33\linewidth}
        \centering
        \includegraphics[width=\linewidth]{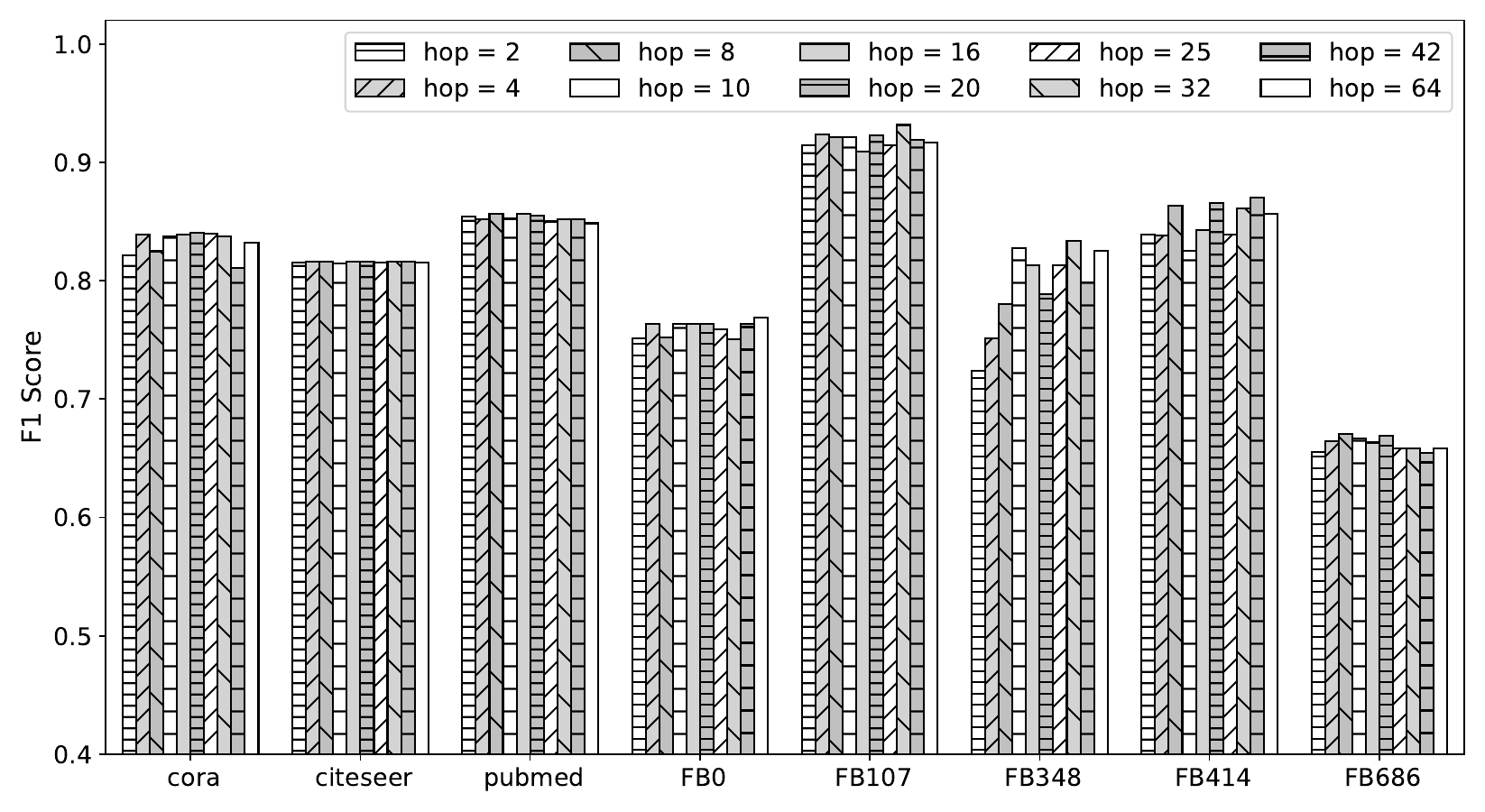}
        \vspace{-6mm}
        \caption{F1-Score by varying hop number}
        \label{fig:chart5}
        \vspace{-6mm}
    \end{subfigure}
    \hfill
    \begin{subfigure}[b][4cm][c]{0.33\linewidth}
        \centering
        \includegraphics[width=\linewidth]{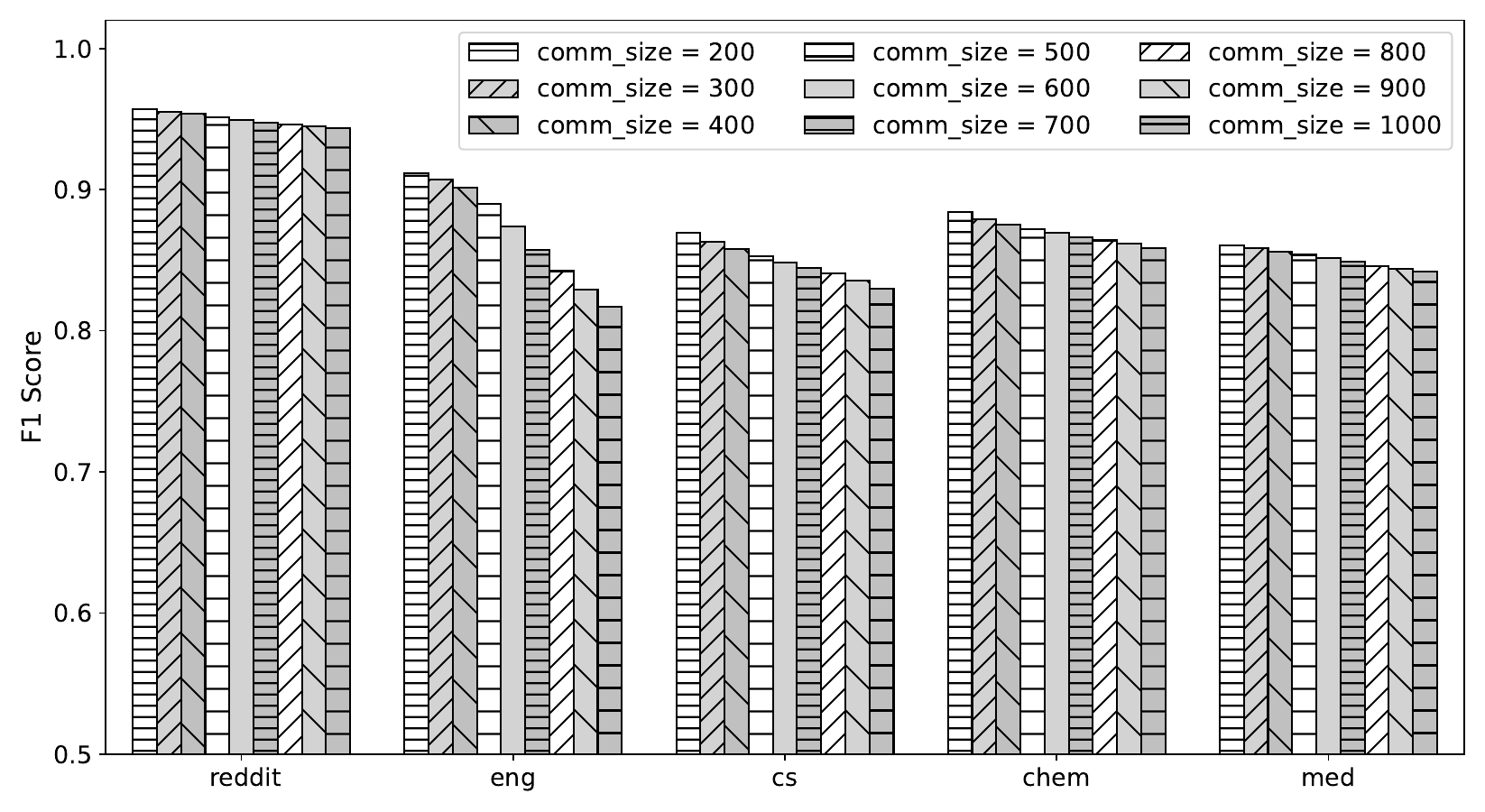}
        \vspace{-6mm}
        \caption{F1-Score by varying community size}
        \label{fig:chart6}
        \vspace{-6mm}
    \end{subfigure}
    \hfill
    \begin{subfigure}[b][4cm][c]{0.33\linewidth}
        \centering
        \includegraphics[width=\linewidth]{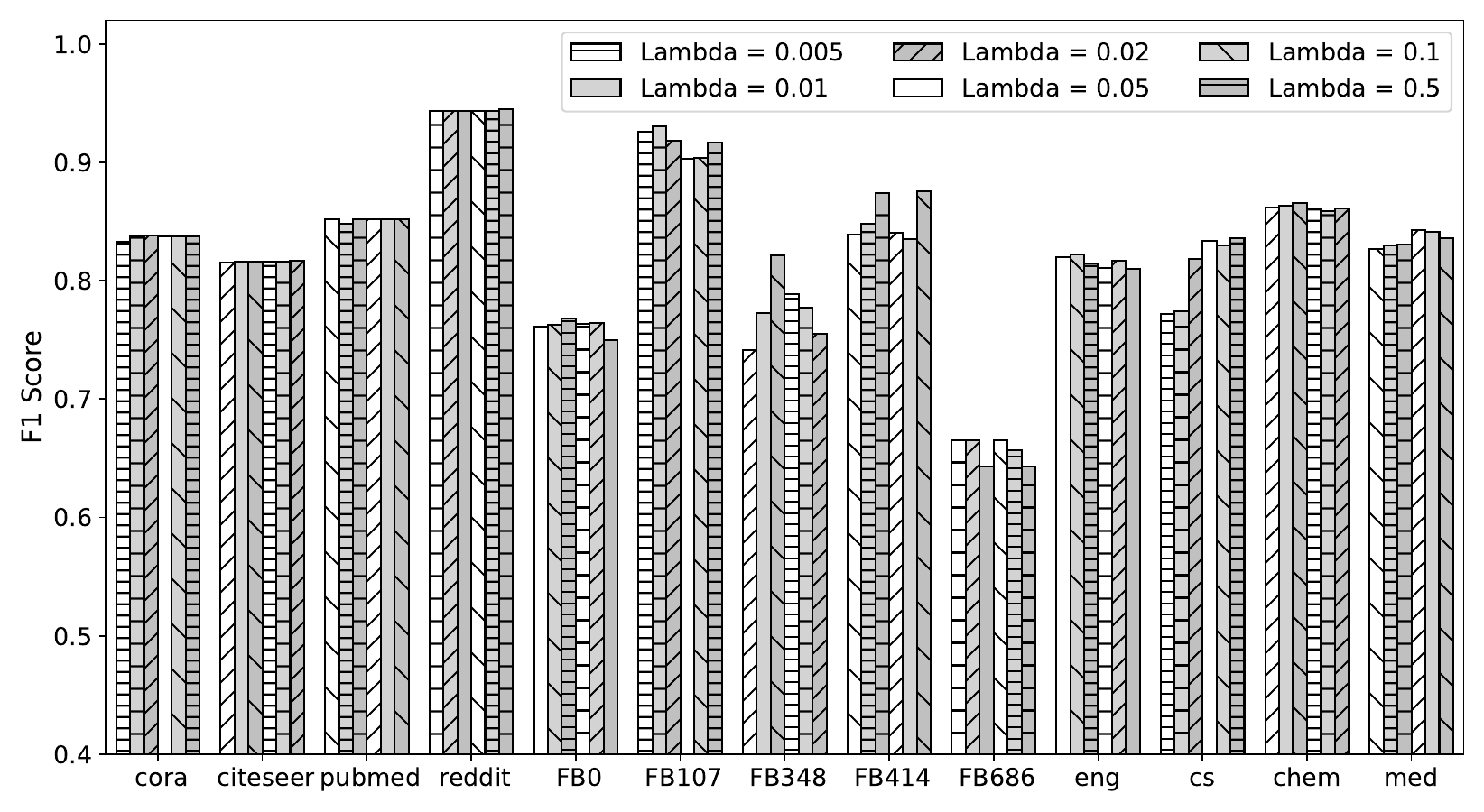}
        \vspace{-6mm}
        \caption{Varying $\lambda$ in the soft sparse filter}
        \label{fig:chart1}
        \vspace{-6mm}
    \end{subfigure}
    \hfill
    \begin{subfigure}[b][4cm][c]{0.33\linewidth}
        \centering
        \includegraphics[width=\linewidth]{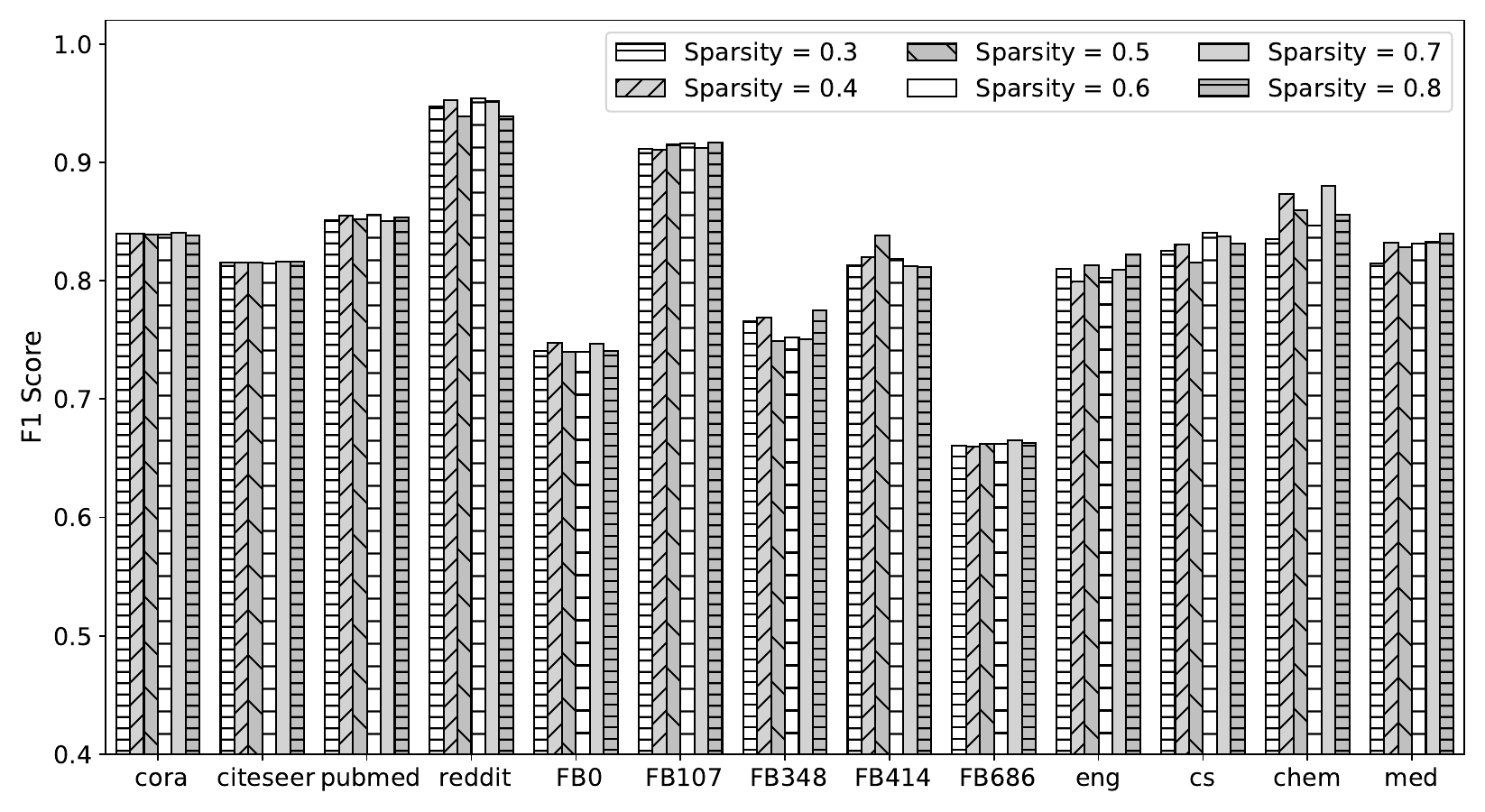}
        \vspace{-6mm}
        \caption{Varying sparsity rate in the hard sparse filter}
        \label{fig:chart2}
    \end{subfigure}
    \hfill
    \begin{subfigure}[b][4cm][c]{0.33\linewidth}
        \centering
        \includegraphics[width=\linewidth]{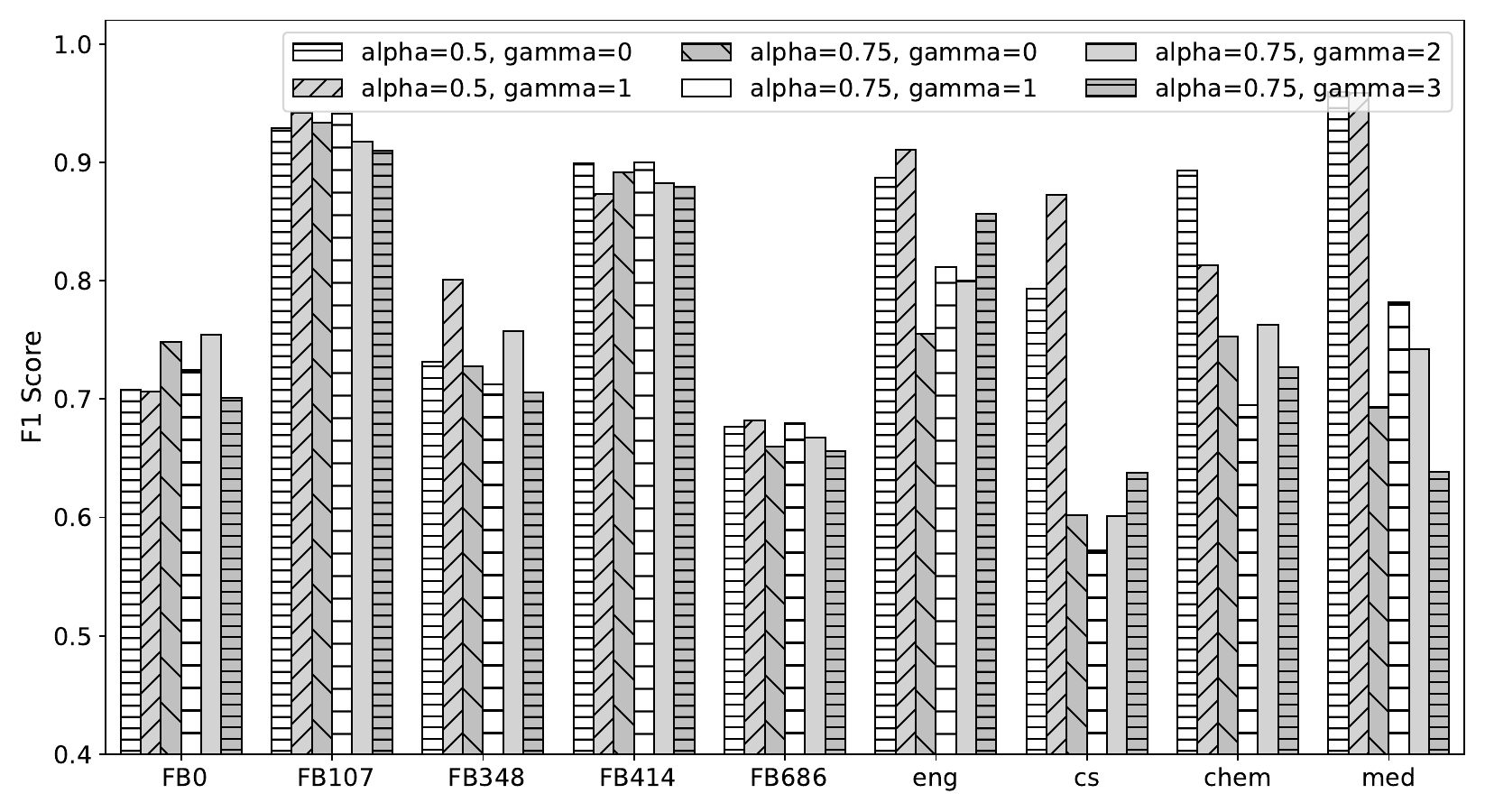}
        \vspace{-6mm}
        \caption{F1-Score by varying $\alpha$ and $\gamma$}
        \label{fig:chart3}
    \end{subfigure}
    \hfill
    \begin{subfigure}[b][4cm][c]{0.33\linewidth}
        \centering
        \includegraphics[width=\linewidth]{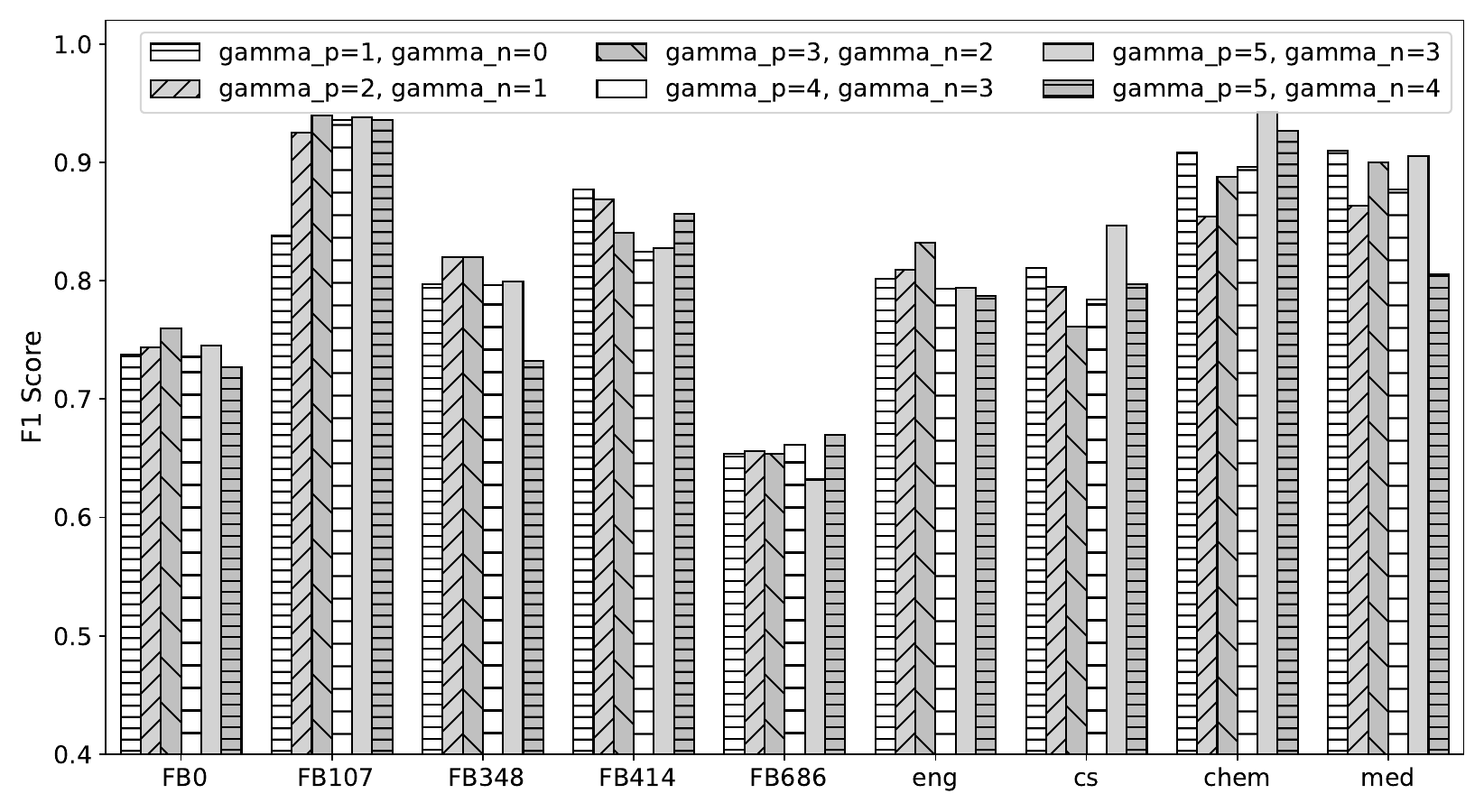}
        \vspace{-6mm}
        \caption{F1-Score by varying $\gamma^+$ and $\gamma^-$}
        \label{fig:chart4}
    \end{subfigure}
    \vspace{-8mm}
    \caption{Hyper-parameter analysis}
    \label{fig:chart}
    \vspace{-5mm}
\end{figure*}

\subsection{Disjoint CS and Ablation Study}
\myparagraph{Disjoint community search} 
\myblue{In this section, we report the model's performance on datasets with disjoint communities to further evaluate its effectiveness.}
The model performance with Sub-CS in disjoint community datasets is presented in \autoref{tab:disjoint}. 
Although SMN is primarily designed for OCS, it outperforms SOTA models in disjoint community search.
SMN consistently achieves superior results compared to all baseline models, showcasing an average improvement of $7.62\%$ across 4 datasets. 
This enhancement is primarily due to the hop-wise attention mechanism and the proposed search algorithm, which leverage high-order patterns captured from a larger model receptive field.
The results prove the effectiveness of SMN in learning representative node embeddings.


\myparagraph{Ablation study}
In this section, we investigate the effectiveness of components employed by SMN and SSF in \autoref{fig:ablation} to illustrate the contribution of each design.
We conduct the ablation study toward three loss functions, SSF, multi-hop attention, and the full model. 
Overall, the full model achieves the most stable and superior performance across 9 datasets, while the model without SSF and the multi-hop attention shows the worst performance (mostly below 0.50 F1-Score). 
We notice that the model without classification loss also performs well in the Facebook datasets. 
This trend is due to Facebook datasets being smaller in size, which makes performing the classification less challenging.
On average, the spatial loss functions and the SSF contribute the most to the model performance, showing improvements of $7.68\%$ and $6.37\%$.

\subsection{Hyper-parameter Analysis}

\myblue{In this section, we conduct various experiments on OCS to test the sensitivity of hyper-parameters.
The study contains six experiments, covering parameters such as the number of hops, community sizes, $\lambda$, $\alpha$, and $\gamma$. 
Hyper-parameters are tuned using grid search.
}

\myparagraph{Varying hop number} In \autoref{fig:chart}(a), we assess the model's performance by varying hop numbers, which determine the model receptive fields. 
\myblue{This experiment evaluates the model's capability of capturing high-order patterns and robustness of oversmoothing.}
SMN demonstrates a stable and slightly increasing trend as the number of hops increases.
Notably, GNN models tend to suffer from oversmoothing, generally limiting the receptive fields to 3 hops.
This proves that SMN benefits from higher-order receptive fields by mitigating the oversmoothing effect.

\myparagraph{Varying community size} In \autoref{fig:chart}(b), we evaluate the effect of varying community sizes on the model's F1-Score. 
The community sizes tested range from 200 to 1000 in increments of 100. 
\myblue{This experiment evaluates the model's sensitivity to community size, where smaller communities represent a less challenging task compared to larger ones.}
The results indicate that the model's performance slightly decreases as the community size increases. 
This demonstrates the model's effectiveness in identifying large communities, as it maintains a high F1-Score even with larger community sizes. 

\myparagraph{Varying $\lambda$ in the soft sparse filter} In \autoref{fig:chart}(c), we evaluate the performance of the model by varying the parameter $\lambda$ in the soft sparse filter. 
\myblue{$\lambda$ is utilized to control the sparsity level in the regulation term in \autoref{eq:l1_penalty}.}
The values of $\lambda$ tested are 0.005, 0.01, 0.02, 0.05, 0.1, and 0.5. 
The figure shows that the impact of $\lambda$ varies across different datasets. 
The model performance is more sensitive to the $\lambda$ values for the datasets with overlapping communities.  


\myparagraph{Varying sparsity rate in the hard sparse filter} \autoref{fig:chart}(d) presents the evaluation results of the model by varying sparsity rates in the hard sparse filter. 
\myblue{Unlike using $\lambda$ to control the sparsity, the hard filter directly zeroes out elements with lower absolute weights in the model classifier based on a predefined sparsity rate.}
Similar to the findings in above, disjoint datasets do not benefit significantly from sparsity adjustments, whereas the optimal sparsity settings in overlapping datasets are highly dataset-specific.

\myparagraph{Varying $\alpha$ and $\gamma$} In \autoref{fig:chart}(e), we analyze the effect of varying $\alpha$ and $\gamma$ on the model's performance. 
\myblue{Here, $\alpha$ controls the balance between different loss components, and $\gamma$ influences the overall weight of the focal losses.}
The results reported for $\alpha$ are 0.5 and 0.75, and for $\gamma$ are 0, 1, 2, and 3. 
This shows that different combinations of $\alpha$ and $\gamma$ yield varying F1-Score across datasets. 

\myparagraph{Varying $\gamma^+$ and $\gamma^-$} \autoref{fig:chart}(f) explores the impact of varying $\gamma^+$ and $\gamma^-$ on the model's F1-Score. 
\myblue{These parameters control the influence of positive and negative samples in the training process as shown in \autoref{eq:asl_loss}.}
The values tested include combinations such as $(\gamma^+ = 1, \gamma^- = 0)$, $(\gamma^+ = 3, \gamma^- = 2)$, and $(\gamma^+ = 5, \gamma^- = 3)$. 
The results indicate that different settings of $\gamma^+$ and $\gamma^-$ can significantly affect performance. 
For example, higher values of both $\gamma^+$ and $\gamma^-$ generally lead to better F1-Score in the MAG-Chem and MAG-CS datasets. 
Conversely, a more balanced setting is preferable in others.


    
\section{Conclusion and Future Work} \label{Conclusions}
This paper studies community search in complex network structures, particularly the challenging domain of overlapping communities (OCS). 
A general solution of OCS named SSF is proposed, accompanied by a Simplified Multi-hop Attention Network (SMN). 
The model enables effective exploration of the overlapping community structure within networks. 
Extensive experiments on 13 real-world datasets prove the superiority of SMN compared to the state-of-the-art approaches across various dimensions, including model effectiveness, training efficiency, and query efficiency.
\myblue{In real-world applications, graphs can be extremely large and evolve over time. 
One potential limitation of our model is its scalability, as computing high-order adjacency matrices is space-intensive. 
Future work could focus on optimizing space complexity and improving scalability for large graphs. 
Additionally, incorporating real-time updates to handle dynamic communities would allow the model to adapt as the network evolves, capturing real-time patterns.}

\bibliographystyle{ACM-Reference-Format}
\bibliography{ref}

\end{document}